\documentclass[lettersize,journal]{IEEEtran}
\usepackage{amsmath,amsfonts}
\usepackage{algorithmic}
\usepackage{algorithm}
\usepackage{array}
\usepackage[caption=false,font=normalsize,labelfont=sf,textfont=sf]{subfig}
\usepackage{textcomp}
\usepackage{stfloats}
\usepackage{url}
\usepackage{verbatim}
\usepackage{graphicx}
\usepackage{cite}
\usepackage{tikz}
\usepackage{float}
\usepackage{listings}
\usepackage{hyperref}
\usepackage{multirow}   
\usepackage{amssymb}
\usepackage{balance}
\usepackage{filecontents}
\usepackage{xspace}
\usepackage{calligra}
\usepackage{graphicx}
\usepackage{amssymb}
\usepackage{tcolorbox}
\usepackage{setspace}
\usepackage{enumitem}
\usepackage{amsthm}
\usepackage{array}
\usepackage{multirow}
\usepackage{bm}
\usepackage{here}
\usepackage{xcolor}
\usepackage{booktabs}
\usepackage{array}
\usepackage{diagbox}

\newtheorem{definition}{\textsc{Definition}}
\newtheorem{theorem}{\textsc{Theorem}}
\usepackage{hyperref}
\usepackage{cleveref}
\usepackage{stmaryrd}
\usepackage{mathrsfs}
\usepackage{tcolorbox}
\usepackage{hyperref}
\usepackage{adjustbox}
\tcbuselibrary{skins, breakable, theorems}

\newtcbtheorem[auto counter]{protocol}{\textbf{Protocol}}
  { enhanced, float,
    colback = white,
    attach boxed title to top left = {yshift = -3mm, xshift = 5mm},
    fonttitle = \sffamily\bfseries\footnotesize,
    left = 2mm,
    right = 2mm
    }{pro}

\newcommand{\SecVPre}{\texttt{IDCloak}\xspace}

\newcommand{\OPPRF}{OPPRF\xspace}
\newcommand{\PPML}{PPML\xspace}
\newcommand{\VPPML}{vPPML\xspace}
\newcommand{\baseline}{\texttt{iPrivJoin}\xspace}
\newcommand{\Peafowl}{\texttt{Peafowl}\xspace}
\newcommand{\share}[1]{\langle {#1} \rangle}

\definecolor{GreenB}{RGB}{0,128,0}

\begin{document}

\title{\SecVPre: A Practical Secure Multi-party Dataset Join Framework for Vertical Privacy-preserving Machine Learning}

\author{
\IEEEauthorblockN{
Shuyu Chen, 
Guopeng Lin, 
Haoyu Niu, 
Lushan Song, 
Chengxun Hong, 
Weili Han,~\IEEEmembership{Member, IEEE} 
}

\thanks{
This work has been submitted to the IEEE for possible publication. Copyright may be transferred without notice, after which this version may no longer be accessible.}
\thanks{
This work was supported in part by the
National Natural Science Foundation of China under Grant 92370120, Grant 62172100. (\textit{Corresponding author:
Weili Han.})
}
\thanks{
Shuyu Chen, Guopeng Lin, Haoyu Niu, Lushan Song, Chengxun Hong, and Weili Han are with the School of Computer Science, Fudan University, Shanghai 10246, China (e-mail: 23110240005@m.fudan.edu.cn; 17302010022@fudan.edu.cn; 23212010019@m.fudan.edu.cn; 19110240022@fudan.edu.cn;
22300240021@m.fudan.edu.cn;
wlhan@fudan.edu.cn).
}
}

\maketitle

\begin{abstract}

Vertical privacy-preserving machine learning (\VPPML) enables multiple parties to train models on their vertically distributed datasets while keeping datasets private. 
In \VPPML, it is critical to perform the secure dataset join, which aligns features corresponding to intersection IDs across datasets and forms a secret-shared and joint training dataset.
However, existing methods for this step could be impractical due to: (1) they are insecure when they expose intersection IDs; or (2) they rely on a strong trust assumption requiring a non-colluding auxiliary server; or (3) they are limited to the two-party setting. 

This paper proposes \SecVPre, the first practical secure multi-party dataset join framework for \VPPML that keeps IDs private without a non-colluding auxiliary server. 
\SecVPre consists of two protocols: (1) a circuit-based multi-party private set intersection protocol (cmPSI), which obtains secret-shared flags indicating intersection IDs via an optimized communication structure combining OKVS and OPRF; (2) a secure multi-party feature alignment protocol, which obtains the secret-shared and joint dataset using secret-shared flags, via our proposed efficient secure shuffle protocol. 
Experiments show that: 
(1) compared to the state-of-the-art secure two-party dataset join framework (\baseline), \SecVPre demonstrates higher efficiency in the two-party setting and comparable performance when the party number increases;
(2) compared to the state-of-the-art cmPSI protocol under honest majority, our proposed cmPSI protocol provides a stronger security guarantee (dishonest majority) while improving efficiency by up to $7.78\times$ in time and $8.73\times$ in communication sizes;
(3) our proposed secure shuffle protocol outperforms the state-of-the-art secure shuffle protocol by up to  $138.34\times$ in time and $132.13\times$ in communication sizes.

\end{abstract}

\begin{IEEEkeywords}
Secure multi-party computation,
private set intersection, secure dataset join
\end{IEEEkeywords}

\section{Introduction}

\IEEEPARstart{P}{rivacy-preserving} machine learning (\PPML) enables multiple parties to cooperatively train machine learning models on their datasets with privacy preservation. 
Among various \PPML paradigms, multi-party 
vertical \PPML (\VPPML), where parties' datasets are vertically distributed, i.e. overlap in sample IDs but have distinct feature sets, has extensive real-world applications across multiple fields such as healthcare, finance~\cite{chen2021homomorphic,lin2024ents,mohassel2017secureml}. Multi-party \VPPML significantly expands the feature dimension, thereby enhancing the performance of the trained model. For instance, in a healthcare \VPPML scenario, there are multiple organizations, e.g., a hospital with patients' health records, a research institute with genetic data, and an insurance company with demographics. They can apply multi-party \VPPML to expand the feature dimension, and jointly train a diagnostic model on their vertically distributed datasets with privacy preservation.

A critical step in \VPPML is the secure dataset join. As is shown in Figure \ref{fig:arch}, this step typically involves two phases: (1) computing the intersection IDs across datasets while keeping IDs private; and (2) aligning features corresponding to these intersection IDs while keeping the datasets private. After the two phases, each party $P_i$ ($i \ge 2$) holds a secret-shared and joint training dataset $\share{\mathcal{D}}_i$ consisting of the aligned features corresponding to the intersection IDs. Then parties can train a machine learning model on $\share{\mathcal{D}}_i$ with privacy preservation.

\begin{figure}[tb]
    \centering
    \includegraphics[width=0.47\textwidth]{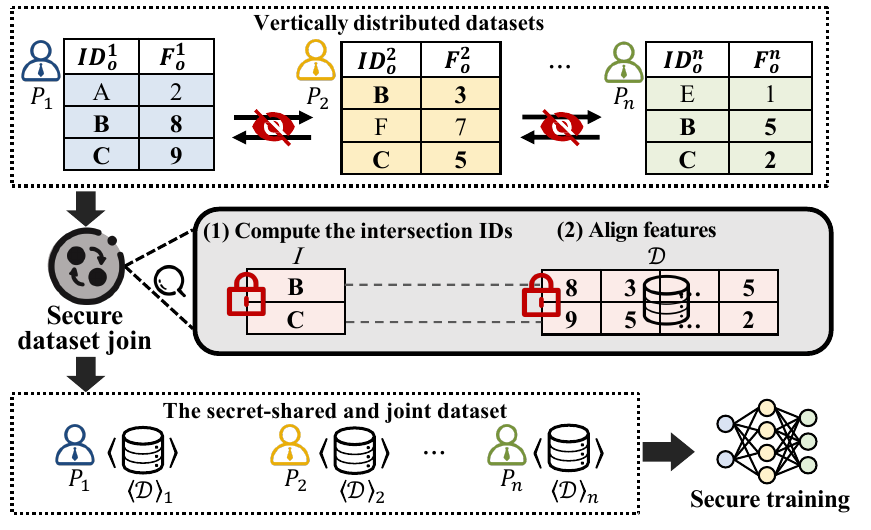}
    \vspace{-1em}
    \caption{Illustration of secure dataset join in vertical \PPML. Multiple parties $P_1, P_2,\ldots, P_n$, sharing the intersection IDs $B$ and $C$.}
    \label{fig:arch}
\end{figure}

\IEEEpubidadjcol

However, there remains a critical gap in providing a practical secure multi-party dataset join framework for \VPPML due to at least one of the following reasons.
(1) Most existing secure multi-party dataset join methods for \VPPML are insecure since these methods expose the intersection IDs~\cite{kolesnikov2017practical,pinkas2018scalable,bay2021multi,oringkstar24}. Specifically, exposing intersection IDs could disclose sensitive information about individuals. For example, in the healthcare scenario above, exposing intersection IDs discloses patient lists of the hospital. Besides, exposing intersection IDs could suffer from reconstruction attacks, potentially recovering the original training data, as highlighted by Jiang \textit{et al.}~\cite{jiang2022comprehensive}.
(2)
Although Gao \textit{et al.} \cite{gao2025peafowl} propose a secure multi-party dataset join method \Peafowl for \VPPML that keeps intersection IDs private, this method relies on a strong trust assumption requiring a non-colluding auxiliary server. This assumption could be impractical in real-world multi-party \VPPML settings, where such server typically does not exist.
(3) 
While Liu \textit{et al.}~\cite{liu2023iprivjoin} propose a secure dataset joins method \baseline for \VPPML that keeps intersection IDs private without requiring a non-colluding auxiliary server, it is limited to the two-party setting.

We note that the secure dataset join in the multi-party setting is notoriously harder to tackle than in the two-party setting:
(1) the multi-party setting introduces difficulty in preserving the privacy of intersection IDs across any subset of parties;
(2) the multi-party setting introduces the risk of collusion among subsets of parties, which is absent with only two parties; 
and (3) the multi-party setting increases considerable costs when supporting three or more parties compared to two-party settings \cite{kolesnikov2017practical}.

As a result, there remains a critical challenge:
\textit{How to achieve practical secure multi-party dataset join for \VPPML without a non-colluding auxiliary server?}

\subsection{Our Approaches}
To address the above challenge, we propose \SecVPre, the first practical secure dataset join framework for multi-party \VPPML that keeps IDs private without a non-colluding auxiliary server. 

\noindent \textbf{High-level Idea.}
To preserve the privacy of intersection IDs across any subset of parties while resisting collusion attacks by up to $n-1$ parties in semi-honest settings, where $n$ is the number of parties, \SecVPre consists of our two proposed protocols.
(1) An efficient circuit-based multi-party private set intersection protocol (cmPSI), which leverages oblivious key-value store (OKVS) and oblivious pseudorandom function (OPRF) to enable each party to obtain secret-shared flags $\share{\Phi_{b \times 1}}_i$ while keeping IDs private, where $b$ is the size of the input dataset. 
Specifically, for $j \in b$, if the party $P_1$'s ID vector $\textit{ID}^1[j]$ is in the intersection, the plaintext flag $\Phi[j]$ is 0, and randomness otherwise.
(2) A secure multi-party feature alignment protocol (smFA), which enables each party to obtain the secret-shared and joint dataset while keeping datasets private. Initially, parties use OKVS and secret sharing to obtain a secret-shared dataset consisting of aligned features and redundant data. Specifically, for $j \in [b]$, if the party $P_1$'s ID vector $\textit{ID}^1[j]$ is in the intersection, the data in the $j$-th row of the dataset is aligned features, and redundant data otherwise.
To remove redundant data, parties execute a secure shuffle protocol on the secret-shared dataset concatenating secret-shared flags $\share{\Phi}$, reconstruct shuffled $\share{\hat{\Phi}}$ into plaintext flags $\hat{\Phi}$, and remove redundant data from the dataset based on whether corresponding entries in $\hat{\Phi}$ are randomness.
As the secure shuffle disrupts the original order of flags $\hat{\Phi}$ and each plaintext value in $\hat{\Phi}$ is 0 or randomness, $\hat{\Phi}$ does not disclose the original IDs themselves, ensuring that no party can infer the original intersection IDs.

To enhance the efficiency while ensuring the security of our \SecVPre, we introduce the following optimizations:
(1) for optimizing the cmPSI protocol,  we propose an optimized communication structure for transmitting OKVS tables, dynamically configured according to parameters, i.e. the number of parties $n$, dataset size $m$, bit length $l$, network bandwidth, and network latency;
and (2) for optimizing the smFA protocol, we propose a novel secure shuffle protocol, which reduces the communication sizes for one party from \(O(ndlm\log m)\) to \(O(ndlm)\) compared to the SOTA, where $d$ is the feature dimension. Since secure multi-party shuffle dominates the communication and time costs (over $99\%$ in our experiments) in smFA, our novel secure shuffle protocol significantly boosts the practicality of smFA in the multi-party setting.

Overall, \SecVPre follows a workflow similar to the SOTA two-party framework \baseline\cite{liu2023iprivjoin}: first hashes the dataset, securely generates a secret-shared dataset that includes aligned features and redundant data, and finally removes the redundant data through a secure shuffle protocol. The main difference is that we designed a multi-party protocol for each step. Additionally, by reusing the results of OPRF and employing a more lightweight OKVS primitive compared to the oblivious programmable pseudorandom function (\OPPRF) used by the \baseline, our \SecVPre achieves lower communication costs and fewer communication rounds, making it more efficient than \baseline in two-party settings.

\subsection{Contributions}
Our contributions are summarized below: 
\begin{itemize}[itemsep=0pt, topsep=0pt, partopsep=0pt, parsep=0pt]
\item  To our best knowledge, we propose the first practical secure multi-party dataset join framework for \VPPML without a non-colluding auxiliary server, \SecVPre, which preserves the privacy of intersection IDs across any subset of parties while resisting collusion attacks by up to $n-1$ parties. 
\item We propose two efficient protocols for \SecVPre: (1) a cmPSI protocol, which obtains secret-shared flags indicating intersection IDs via an optimized communication structure; and (2) a secure multi-party feature alignment protocol, which obtains the secret-shared and joint dataset using the secret-shared flags via our proposed efficient secure multi-party shuffle protocol. 
\end{itemize}

We evaluate \SecVPre across various party numbers ($2 \sim 6$) using six real-world datasets, with feature dimensions between $10$ and $111$ and total dataset sizes between $1353$ and $253680$. The experimental results show that: (1) in the two-party setting, \SecVPre outperforms the SOTA secure two-party dataset join framework \baseline~\cite{liu2023iprivjoin} by $1.69\times \sim 1.92\times$ and $1.50\times \sim 1.72\times$ in terms of time and communication sizes, respectively. Meanwhile \SecVPre still achieves comparable efficiency to the two-party \baseline even as the number of parties increases; 
(2) our proposed cmPSI protocol with a stronger security guarantee (resists against up to \(n-1\) colluding parties) outperforms the SOTA cmPSI protocol~\cite{chandran2021circuit} (resists against up to \(n/2-1\) colluding parties) by $1.39\times \sim 7.78\times$ and $5.58\times \sim 8.73\times$ in terms of time and communication sizes, respectively. 
(3) our proposed secure multi-party shuffle protocol outperforms the SOTA secure shuffle protocol\cite{mp-spdz} by $21.44\times \sim 138.34\times$ and $105.69\times \sim 132.13\times$ in terms of time and communication sizes, respectively.
\section{Related work}

\subsection{Private Set Intersection}

\subsubsection{Multi-party PSI (mPSI)}
mPSI enables multiple parties to compute the intersection of their ID sets while keeping non-intersection IDs private. 
In mPSI-based secure dataset join, parties use mPSI to obtain intersection IDs in plaintext. They can then locally sort the features according to the lexicographical order of these intersection IDs and secretly share these ordered features with other parties. By merging each row of these features, the parties can obtain a secret-shared and joint dataset.
Despite significant advances in efficient mPSI protocols~\cite{kolesnikov2017practical,inbar2018efficient,kavousi2021efficient,bay2021practical,bay2021multi,wu2024ring}, existing mPSI solutions usually expose intersection elements, leading to privacy leakage.

\subsubsection{Circuit-based PSI (cPSI)}
The initial cPSI protocol introduced by Huang \textit{et al.} ~\cite{huang2012private} allows two parties to securely obtain the secret-shared intersection of their ID sets while keeping the IDs private. Subsequent research on two-party cPSI~\cite{rindal2021vole,garimella2021oblivious,chandran2021circuit,rindal2022blazing,ma2022secure} explores leveraging OPPRF or private set membership techniques to extend cPSI for obtaining secret-shared aligned features corresponding to intersection IDs, making it suitable for secure dataset join in \VPPML.
However, cPSI outputs a joint dataset that includes both features corresponding to intersection IDs and redundant data (i.e. secret-shared zeros). According to the literature~\cite{liu2023iprivjoin}, training on the joint dataset output by cPSI can lead to a $3\times$ increase in both training time and communication sizes, compared with training on the joint dataset without redundant data.

\subsubsection{Circuit-based Multi-party PSI (cmPSI)}
Existing cmPSI protocols ~\cite{li2021prism,chandran2021efficient,ma2022secure} enable parties to obtain secret-shared intersection IDs but do not consider obtaining aligned features corresponding to the intersection IDs.

\subsection{Secure Two-party Dataset Join}
Liu \textit{et al.} \cite{liu2023iprivjoin} propose a secure two-party dataset join framework for \VPPML, \baseline. Similar to cPSI, \baseline first employs the OPPRF to construct a secret-shared and joint dataset that includes both aligned features corresponding to intersection IDs and redundant data. To eliminate redundant data, \baseline utilizes a secure shuffle protocol. However, \baseline is inherently designed for the two-party setting.

\subsection{Secure Multi-party Dataset Join with Non-colluding Auxiliary Server}

Gao \textit{et al.}~\cite{gao2025peafowl} propose \Peafowl, a secure multi-party dataset join solution, but \Peafowl requires an auxiliary server to protect intersection IDs and assumes this auxiliary server never colludes with any parties, which is a strong trust assumption. This assumption could be impractical in real-world multi-party \VPPML settings, where such server typically does not exist.
In contrast to \Peafowl, our \SecVPre does not require a non-colluding auxiliary server, thereby achieving stronger security guarantees.

As is shown in Table~\ref{tab:techcmp}, we summarize the strengths and limitations of the representative secure dataset join methods.
\begin{table}[htbp]
\belowrulesep=0pt 
\aboverulesep=0pt
\caption{Comparison of various secure dataset join methods. Here, `Parties' denotes the number of supported parties, `ID-Private' denotes whether intersection IDs are kept private, `No Redundant' denotes whether redundant data is avoided in the joint dataset, and `No Server' denotes whether a non-colluding auxiliary server is not required. The better settings within each column are highlighted in green.}
\centering
\scalebox{0.82}{
\setlength{\tabcolsep}{3pt}
\begin{tabular}{c|cccc}
\toprule
Method    & Parties           & ID-Private & No Redundancy & No Server \\ \midrule
mPSI      & \textcolor{GreenB}{\textbf{$n\ge2$}} & no         & \textcolor{GreenB}{\textbf{yes}}    & \textcolor{GreenB}{\textbf{yes}}     \\
cPSI      & $n=2$               & \textcolor{GreenB}{\textbf{yes}}        & no  & \textcolor{GreenB}{\textbf{yes}}      \\
\Peafowl & \textcolor{GreenB}{\textbf{$n\ge2$}}               & \textcolor{GreenB}{\textbf{yes}}       & \textcolor{GreenB}{\textbf{yes}}    & no     \\
\baseline & $n=2$               & \textcolor{GreenB}{\textbf{yes}}       & \textcolor{GreenB}{\textbf{yes}}    & \textcolor{GreenB}{\textbf{yes}}     \\
Ours      & \textcolor{GreenB}{\textbf{$n\ge2$}} & \textcolor{GreenB}{\textbf{yes}}        & \textcolor{GreenB}{\textbf{yes}}  & \textcolor{GreenB}{\textbf{yes}}
\\ \bottomrule
\end{tabular}
}
\label{tab:techcmp}
\end{table}
\section{Overview of \SecVPre}

We summarize the frequently used notations in Table~\ref{tab:notations}.

\begin{table}[htbp]
\belowrulesep=0pt 
\aboverulesep=0pt
\centering
\caption{Notation Table.}
\scalebox{0.82}{
\setlength{\tabcolsep}{1pt}
\renewcommand{\arraystretch}{0.95}
\begin{tabular}{c|l}
\toprule
Symbol & Description
\\
\midrule
$n$ & The number of parties.
\\
$[x_1,x_2]$ & \begin{tabular}[c]{@{}c@{}} The set $\{x_1,...,x_2\}$. \end{tabular}
\\
$[x]$ & \begin{tabular}[c]{@{}c@{}} The set $\{1,2,...,x\}$. \end{tabular}
\\
$P_i$  &
\begin{tabular}[c]{@{}c@{}} The $i$-th party for $i \in [n]$.\end{tabular}
\\
$X^i$ & \begin{tabular}[c]{@{}c@{}} The data belonging to $P_i$. \end{tabular}
\\
${\textit{ID}^i}$  &
\begin{tabular}[c]{@{}c@{}} The IDs in hash table of $P_i$.
\end{tabular}
\\
${F^i}$  & 
\begin{tabular}[c]{@{}c@{}} The features corresponding to ${\textit{ID}^i}$.
\end{tabular}
\\
$I_{id}$ & \begin{tabular}[c]{@{}c@{}} The intersection IDs ($I_{id} = \bigcap_{i=1}^n \textit{ID}^i$).  \end{tabular}
\\
$\Phi$ & \begin{tabular}[c]{@{}c@{}} The flags indicating intersection IDs $I_{id}$.  \end{tabular}
\\
$\mathcal{D}$ & \begin{tabular}[c]{@{}c@{}} The dataset consisting of aligned features corresponding to $I_{id}$. \end{tabular}
\\
$m$ & \begin{tabular}[c]{@{}c@{}}The size of dataset held by each party.\end{tabular}
\\
$b$ & \begin{tabular}[c]{@{}c@{}}The number of bins in hash table.\end{tabular}
\\
$d_i$ & \begin{tabular}[c]{@{}c@{}}
The feature dimension of datasets held by $P_i$.\end{tabular}
\\
$d$ & \begin{tabular}[c]{@{}c@{}}
The total feature dimension ($d = \sum_{i=1}^{n}d_i$).\end{tabular}
\\
$c$ & \begin{tabular}[c]{@{}c@{}} The size of $I$ or $\mathcal{D}$.  \end{tabular}
\\
$h$ & \begin{tabular}[c]{@{}c@{}} The number of hash functions.  \end{tabular}
\\
$\mathbb{Z}_{2^l}$ & \begin{tabular}[c]{@{}c@{}} Ring of size $l$ bits; $l = 64$ in this paper.\end{tabular}
\\
$\share{x}_i$ & \begin{tabular}[l]{@{}l@{}} The secret-shared value  of $x \in \mathbb{Z}_{2^{l}}$ held by $P_i$
\\
s.t. $x=\sum_{i=1}^n \share{x}_i$.
\end{tabular}
\\
$x\|y$ & \begin{tabular}[c]{@{}c@{}}The concatenation of $x$ and $y$.\end{tabular}
\\
$X\|Y$ & \begin{tabular}[c]{@{}c@{}}The row-by-row concatenation of $X$ and $Y$.\end{tabular}
\\
$|X|$ & \begin{tabular}[c]{@{}c@{}}The size of $X$.\end{tabular}
\\
$X_{a\times b}$ & \begin{tabular}[c]{@{}c@{}}The matrix $X$ with size $a\times b$.\end{tabular}
\\
$X[i]$ & \begin{tabular}[c]{@{}c@{}} The $i$-th element or $i$-th row of $X$. \end{tabular}
\\
$X[i][j]$ & \begin{tabular}[c]{@{}c@{}} The $j$-th element in the $i$-th row of $X$. \end{tabular}
\\
$\kappa$ & \begin{tabular}[c]{@{}c@{}}The computational security parameter.
\end{tabular}
\\
$\lambda$ & \begin{tabular}[c]{@{}c@{}}The statistical security parameter.
\end{tabular}
\\
\bottomrule
\end{tabular}
}
\label{tab:notations}
\end{table}

\subsection{Problem Statement}
As is shown in figure~\ref{funct:secdjoin}, the secure multi-party dataset join functionality \(\mathcal{F}_{\textit{smDJoin}}\) takes as input from each party \(P_i\) (\(i \in [n]\)) a dataset \((\textit{ID}_o^i \| F_o^i)\), where \(\textit{ID}_o^i\) is an \(m\)-element vector of IDs and \(F_o^i\) is an \(m \times d_i\) feature matrix. 
Here, “\(\|\)” indicates the row-by-row concatenation of sample IDs and their corresponding features.
Once \(\mathcal{F}_{\textit{smDJoin}}\) is executed, each \(P_i\) obtains a secret-shared and joint dataset \(\share{\mathcal{D}}_i\) that consists of the aligned features corresponding to the intersection IDs.

\begin{figure}[htbp]
    \centering
    \setlength{\fboxsep}{3pt}
    \fbox{
    \scalebox{0.85}{
        \parbox{0.52\textwidth}{
        {\centering \textbf{\underline{Functionality $\mathcal{F}_{ \textit{smDJoin}}$}}\par}
        \medskip
        
        \textbf{Parameters:}
        The number of parties $n$, the size of the dataset held by each party $m$,
        the feature dimension of the dataset held by each party $d_i$ ($i \in [n]$), the total feature dimension $d$,
        and the bit length of element $l$.
        
        \textbf{Inputs:} For each $i \in [n]$, $P_i$ holds dataset $\textit{ID}_o^i||F_o^i$.
        
        \textbf{Functionality:}
        \begin{enumerate}[label=\arabic{enumi}.,leftmargin=1.2em, itemsep=2pt, topsep=2pt, parsep=0pt]
            \item Compute intersection
                      $I_{id} = \bigcap_{i=1}^{n} \textit{ID}^i_o$.
                  Let $c = |I_{id}|$.
            \item Set $\mathcal{D}[j] = \{ F^1_o[j_1][1],\ldots,F^1_o[j_1][d_1],\ldots,F^n_o[j_n][1],\\\ldots,F^n_o[j_n][d_n] \}_{\forall j \in [c]}$, where $F^i_o[j_i][u]$ for $u \in [d_i]$ is the corresponding feature value of $\textit{ID}[j_i] = I[j]$ from $P_i$ for $j_i \in [m]$, $i \in [n]$.
            
            \item Sample $\share{\mathcal{D}}_i \leftarrow \mathbb{Z}_{2^l}^{c \times d}$ such that $\sum_{i=1}^n\share{\mathcal{D}}_i = \mathcal{D}$.
            
            \item Return $\share{\mathcal{D}}_i$ to $P_i$ for $i \in [n]$.
        \end{enumerate}
        }
    }
    }
    \caption{Ideal functionality of secure dataset join.}
    \label{funct:secdjoin}
\end{figure}


\subsection{Security Model}
In this paper, we adopt a semi-honest security model with a dishonest majority. That is, a semi-honest adversary $\mathcal{A}$ can corrupt up to $n-1$ parties and aims to learn extra information from the protocol execution while correctly executing the protocols. 
We establish semi-honest security in the simulation-based model, with our construction involving multiple sub-protocols described using the hybrid model~\cite{canetti2001universally}. 
We give the formal security definition as follows:
\begin{definition}
[Semi-honest Model]
Let $\text{view}^\Pi_C(x,y)$ be the views (including the input, random tape, and all received messages) of $C$ in a protocol $\Pi$, where $C$ is the set of corrupted parties, $x$ is the input of $C$ and $y$ is the input of the uncorrupted party. Let $out(x,y)$ be the protocol's output of all parties and $\mathcal{F}(x,y)$ be the functionality's output. $\Pi$ is said to securely compute a functionality $\mathcal{F}$ in the semi-honest model if for any adversary $\mathcal{A}$ there exists a simulator $\text{Sim}_C$ such that for all inputs $x$ and $y$,
$$
\{\text{view}^\Pi_C(x,y), \text{out}(x,y)\} \approx_c \{\text{Sim}_C(x, \mathcal{F}_C(x,y)), \mathcal{F}(x,y)\}.
$$
\end{definition}
\section{Preliminary}

\subsection{Cuckoo Hashing \& Simple Hashing}
\label{subsec:Cuckoohash}

Cuckoo hashing~\cite{pagh2001Cuckoo} relies on $h$ hash functions, denoted $\{H_k: \{0,1\}^\sigma \mapsto [b]\}_{k \in [h]}$, to map each of $m$ elements $\{x_j\}_{j \in [m], x_j \in \mathbb{Z}_{2^\sigma}}$ to \textit{one} of $h$ potential bins $\{H_k(x_j)\}_{ \forall k \in [h]}$, where $b = \omega m$.  
While cuckoo hashing ensures that each bin holds at most one element, some elements may not find an available bin. Traditional cuckoo hashing employs a stash for these overflow elements. However, Pinkas \textit{et al.}~\cite{pinkas2018scalable} introduce a stash-free variant of Cuckoo hashing that removes the need for this storage. Their findings indicate that with $h=3$ hash functions and a table size $b=1.27m$, the probability of failure for an element failing to find an available bin is at most $2^{-40}$. We utilize these parameters to implement Cuckoo hashing without a stash.

In contrast, simple hashing uses $h$ hash functions  $\{H_k: \{0,1\}^\sigma \mapsto [b]\}_{k \in [h]}$ to map each of $m$ elements $\{x_j\}_{j \in [m], x_j \in \mathbb{Z}_{2^\sigma}}$ into \emph{all} bins $H_k(x_j)_{\forall k \in [h]}$. Thus, unlike Cuckoo hashing, simple hashing allows each bin to hold multiple elements.

\subsection{Oblivious Key-Value Store}
Oblivious key-value store~\cite{garimella2021oblivious} (OKVS) is designed to encode $m$ pairs of key-value pairs where the values are random, such that an adversary cannot determine the original input keys from the encoding. This makes the encoding oblivious to the input keys. The definition is as follows:
\begin{definition} [Oblivious Key-Value Store]
An OKVS is parameterized by a key universe $\mathcal{K}$, a value universe $\mathcal{V}$, input length $m$, output length $m'$ and consists of the two functions: 
\begin{flushleft}
\begin{itemize}
    \item $\mathcal{S} = \textit{Encode}(I)$: The encode algorithm receives a set of $m$ key-value pairs $I = \{(k_i, v_i) \}_{\forall i \in [m], (k_i, v_i) \in (\mathcal{K} \times \mathcal{V}) }$ and outputs the encoding $\mathcal{S} \in \mathcal{V}^{m'} \cup \{\perp\}$.
    \item $v = \textit{Decode}(\mathcal{S}, k)$: The decode algorithm receives the encoding $\mathcal{S} \in \mathcal{V}^{m'}$, a key $k \in \mathcal{K}$ and outputs the associated value $v \in \mathcal{V}$.
\end{itemize}
\end{flushleft}
\end{definition}
An OKVS is computationally oblivious\cite{garimella2021oblivious}, if for any two sets of $m$ distinct keys $\{k_i\}_{\forall i \in [m], k_i \in \mathcal{K}}$ and $\{k'_i\}_{\forall i \in [m],k'_i \in \mathcal{K}}$ and $m$ values $\{v_i\}_{i \in [m]}$ each drawn uniformly at random from $\mathcal{V}$, a computational adversary is not able to distinguish between $S =\textit{Encode}(\{(k_i, v_i)\}_{i \in [n]})$ and $S' =\textit{Encode}(\{(k'_i, v_i)\}_{i \in [m]})$.
Additionally, OKVS has the doubly oblivious~\cite{garimella2021oblivious,rindal2022blazing,rindal2021vole} property where the output of the encoding is a uniformly random element from $\mathcal{V}^{m'}$. In addition, another critical property is random decoding, where the decoded value for any non-input key must be indistinguishable from a uniformly random element from $\mathcal{V}$. These properties are particularly useful for our \SecVPre.

\subsection{Oblivious Pseudorandom Function}
As is shown in Figure~\ref{funct:oprf}, oblivious pseudorandom function (OPRF)~\cite{freedman2005keyword} allows $\mathcal{S}$ to learn a pseudorandom function (PRF) key $k$, while $\mathcal{R}$ learns $F_k(x_1),\dots, F_k(x_{\beta})$ for the inputs $(x_1,\dots,x_{\beta})$. 
During the OPRF, neither $\mathcal{S}$ nor $\mathcal{R}$ can learn any additional information. Specifically, $\mathcal{S}$ does not learn any input $x_i (i \in [\beta])$ from $\mathcal{R}$ and $\mathcal{R}$ does not learn PRF key $k$.

\begin{small}
\begin{figure}[h]
    \centering
    \setlength{\fboxsep}{3pt}
    \fbox{
        \scalebox{0.85}{
        \parbox{0.52\textwidth}{
        {\centering \textbf{\underline{Functionality $\mathcal{F}_{\textit{OPRF}}$}}\par}
        \smallskip
        
        \textbf{Parameters:} The number of elements $\beta$ and the bit length of element $l$. A PRF $F_{(\cdot)}(\cdot)$.

        \textbf{Receiver's inputs:} $\{x_i\}_{\forall i\in [\beta], x_i \in \{0,1\}^l}$

        \textbf{Funcitionality:}
        \begin{enumerate}[leftmargin=1em,itemsep=0pt, topsep=0pt, partopsep=0pt, parsep=0pt]
        \item[-] Sender $\mathcal{S}$ obtains the PRF key $k$.
        \item[-] Receiver $\mathcal{R}$ obtains $\{F_k(x_i)\}_{\forall i \in [\beta]}$.
        \end{enumerate}
        }
    }
    }
    \caption{Ideal functionality of OPRF}
    \label{funct:oprf}
\end{figure}
\end{small}

\subsection{Additive Secret Sharing}
Additive secret sharing (ASS)\cite{mohassel2017secureml,demmler2015aby} is a cryptographic technique that splits a private value into multiple shares, allowing the original value to be reconstructed by summing these shares. It is formally defined as follows:
\begin{itemize}[leftmargin=1em,itemsep=0pt, topsep=0pt, partopsep=0pt, parsep=0pt]
\item \textbf{Secret-shared values}: an \( l \)-bit value \( x \) is additively secret-shared among \( n \) parties as shares \( \share{x}_1, \dots, \share{x}_n \), where each \( \share{x}_i \in \mathbb{Z}_{2^l} \) for $i \in [n]$. The value \( x \) can be reconstructed as:  
$
\sum_{i=1}^{n} \share{x}_i \equiv x \pmod{2^l}$.
\item  \textbf{Sharing $\mathit{Shr}(x, i, Q)$}: party $P_i$ randomly samples $r_j \in \mathbb{Z}_{2^l}$ for all $j \in [|Q|-1]$, sets its own share $\share{x}_i = x - \sum_{j=1}^{|Q|-1} r_j$ and sends $r_j$ to party $P_{Q[j]}$, who sets $\share{x}_{Q[j]} = r_j$.
\item \textbf{Reconstruction $\mathit{Rec}(x)$}: each party $P_i$ ($\forall i \in [n]$) sends its share $\share{x}_i$ to all other parties. Once all shares are gathered, each party $P_i$ ($\forall i \in [n]$) reconstruct the original value by computing $x = \sum_{i=1}^n \share{x}_i$.
\end{itemize}

\subsection{Oblivious Shuffle \& Secure Shuffle }
As is shown in Figure~\ref{func:oshuffle}, the oblivious shuffle allows one party \(P_1\) who holds an matrix \(X_{m\times d}\) and receives \(\share{\hat{X}}_1\), while the other party \(P_2\) who holds a random permutation \(\pi:[m]\rightarrow[m]\) and obtains \(\share{\hat{X}}_2\), ensuring that \(\hat{X} = \pi(X)\). During this oblivious shuffle process, the permutation provider \(P_2\) learns nothing about \(X\), and the data provider \(P_1\) learns nothing about \(\pi\). In this paper, we adopt the oblivious shuffle protocol $\Pi_{\rm {O-Shuffle}}$ proposed in \cite{liu2023iprivjoin}, which realizes \(\mathcal{F}_{\textit{O-Shuffle}}\) with the same security guarantees as ours.

\begin{figure}[h]
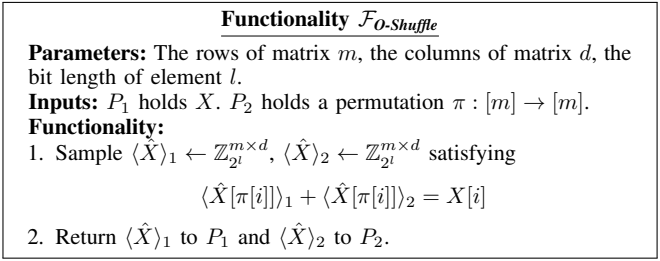

\centering
\fbox{
\scalebox{0.85}{
\parbox{0.52\textwidth}{
\setstretch{0.92}
        {\centering \textbf{\underline{Functionality $\mathcal{F}_{\textit{O-Shuffle}}$}}\par}
        \smallskip

\noindent\textbf{Parameters:} The rows of matrix $m$, the columns of matrix $d$, the bit length of element $l$.

\noindent\textbf{Inputs:} $P_1$ holds $X$. $P_2$ holds a permutation $\pi:[m]\rightarrow[m]$.

\noindent\textbf{Functionality:}
\begin{enumerate}[label=\arabic{enumi}.,leftmargin=1.2em,itemsep=0pt, topsep=0pt, partopsep=0pt, parsep=0pt]
    \item Sample $\share{\hat{X}}_1 \leftarrow \mathbb{Z}_{2^{l}}^{m\times d}$, $\share{\hat{X}}_2 \leftarrow \mathbb{Z}_{2^{l}}^{m\times d}$ satisfying
    \[
    \share{\hat{X}[\pi[i]]}_1 + \share{\hat{X}[\pi[i]]}_2 = X[i]
    \] 
    \item Return $\share{\hat{X}}_1$ to $P_1$ and $\share{\hat{X}}_2$ to $P_2$.
\end{enumerate}
}
}
}
\caption{Ideal functionality of oblivious shuffle}
\label{func:oshuffle}
\end{figure}

As is shown in Figure~\ref{func:shuffle}, the secure shuffle allows the random permutation of share matrix $\share{X}$ from all parties, resulting in refreshed secret-shared matrix $\share{X'} = \share{\pi(X)}$, while keeping the permutation $\pi$ unknown to the parties. 
\begin{figure}[h]
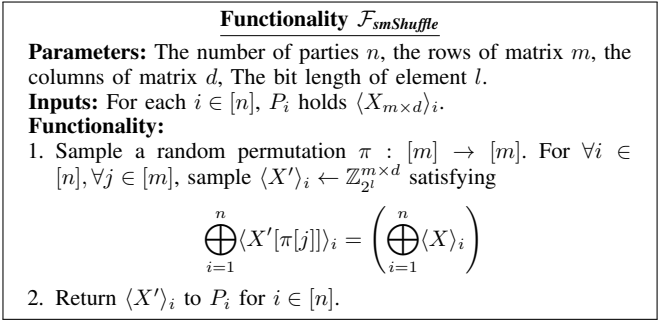

\centering
\fbox{
\scalebox{0.85}{
\parbox{0.52\textwidth}{
\setstretch{0.92}
        {\centering \textbf{\underline{Functionality $\mathcal{F}_{\textit{smShuffle}}$}}\par}
        \smallskip

\noindent\textbf{Parameters:} The number of parties $n$, the rows of matrix $m$, the columns of matrix $d$, The bit length of element $l$.

\noindent\textbf{Inputs:} For each $i \in [n]$, $P_i$ holds $\share{X_{m \times d}}_i$.

\noindent\textbf{Functionality:}
\begin{enumerate}[label=\arabic{enumi}.,leftmargin=1.2em,itemsep=0pt, topsep=0pt, partopsep=0pt, parsep=0pt]
    \item Sample a random permutation $\pi : [m] \rightarrow [m]$. For $\forall i \in [n], \forall j \in [m]$, sample $\share{X'}_i \leftarrow \mathbb{Z}_{2^{l}}^{m\times d}$ satisfying
    \[
    \bigoplus_{i=1}^n \share{X'[\pi[j]]}_i = \left( \bigoplus_{i=1}^n \share{X}_i \right)
    \] 
    \item Return $\share{X'}_i$ to $P_i$ for $i \in [n]$.
\end{enumerate}
}
}
}
\caption{Ideal functionality of secure multi-party shuffle}
\label{func:shuffle}
\end{figure}
\section{Design}
\label{sec:design}

\subsection{Setup Phase}
\label{subsec:off}

During the setup phase, a randomly selected party, denoted as $P_1$, applies stash-less cuckoo hashing on its dataset $\textit{ID}_o^1\|F_o^1$, while the remaining parties, denoted as $P_i(i\ge 2)$ apply simple hashing on their respective dataset $\textit{ID}_o^i\|F_o^i$. Each party maps its dataset based on the unique $id \in \textit{ID}_o^i$ using $h$ public hash functions $\{H_k: \{0,1\}^l \mapsto [b]\}_{k \in [h]}$ along with a hash table ($B_1,\ldots,B_b$). The specific hashing operations are as follows:
\begin{itemize}[leftmargin=1em,itemsep=0pt, topsep=0pt, partopsep=0pt, parsep=0pt]
    \item $P_1$ applies cuckoo hashing using $h$ hash functions $\{H_k: \{0,1\}^l \mapsto [b]\}_{k \in [h]}$, where $b = \omega m$ and $\omega > 1$, that maps each row $(id_j^1,f_{j,1}^1,\ldots,f_{j,d_1}^1) \in \textit{ID}^1_o\|F^1_o$ for $j\in[m]$ to \textit{one of the bins} $\{B_{H_k(id_j^1)}\}_{k\in[h]}$. After hashing, $P_1$ obtains $\textit{ID}^1 \| F^1$, where $\textit{ID}^1[j] = (id^1\|j) $ for $ j \in \{H_k(id^1)\}_{k\in[h]} \cap id^1 \in \textit{ID}^1_o $ and ${F^1}[j]$ is the feature matrix corresponding to $\textit{ID}^1$. Since $b > m$ (where $m = |\textit{ID}^1_o|$), $P_1$ fills each empty $j$-th bin with uniformly random values: 
    $\textit{ID}^1[j] = r_x\|r_i$ where $r_x \in \mathbb{Z}_{2^{l}}$, $r_i \in \mathbb{Z}_{2^l}\setminus [b]$ and $F^1[j] = (r_1,\ldots,r_{d_1})$ where $r_1,\ldots,r_{d_1} \in \mathbb{Z}_{2^l}$.
    \item Each $P_i (i \ge 2)$ applies simple hashing using $h$ hash functions $\{H_k: \{0,1\}^l \mapsto [b]\}_{k \in [h]}$ that maps each row $(id_j^i, f_{j,1}^i, \ldots, f_{j,d_i}^i) \in \textit{ID}^i_o\|F^i_o$ for $j\in[m]$ to \textit{all bins} $\{B_{H_k(id_j^i)}\}_{k\in[h]}$. After hashing, $P_i$ obtains $\textit{ID}^i \| F^i$, where $\textit{ID}^i[j] = \{ (id^i\|j)\}_{id^i \in \textit{ID}_o^i, j \in \{H_1(id^i), \ldots,H_h(id^i) \} }$ and $F^i$ represents the features corresponding to $\textit{ID}^i$.
\end{itemize}

In particular, for cuckoo hashing, we follow the approach in~\cite{pinkas2018scalable,demmler2018pir}, which guarantees that each element \( x \in X \) is placed in the hash table with an overwhelming probability. After hashing, each bin in \( P_1 \)’s hash table holds exactly one unique transformed ID, whereas each bin in \( P_i \)’s (for \( i \ge 2 \)) hash table may contain zero or multiple transformed IDs. Furthermore, rather than storing the original ID directly, parties store its concatenation with the bin index \( j \) (i.e., \( id \| j \)).

\subsection{Private ID Intersection}
\label{subsec:idintersection}

\subsubsection{Definition}
As is shown in Figure~\ref{funct:cmpsi}, 
functionality $\mathcal{F}_{\textit{cmPSI}}$ enables each party $P_i$ ($i \in [n]$) to hold $\textit{ID}^i$ processed by hashing and to obtain secret-shared flags $\share{\Phi_{b\times1}}_i$. 
For each $j \in [b]$, $\Phi[j]$ ($= \sum_{i=1}^n\share{\Phi}_i$) indicates whether $id_j (= \textit{ID}^1[j])$ is an element in intersection IDs $I_{id}$. In particular, $\Phi[j]$ is 0 if $id_j \in I_{id}$. Or, $\Phi[j]$ is a randomness if $id_j \notin I_{id}$.

\begin{figure}[h]
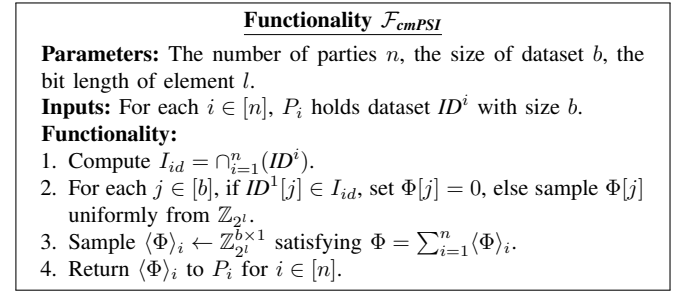

    \centering
    \fbox{
    \scalebox{0.85}{
        \parbox{0.52\textwidth}{
        {\centering \textbf{\underline{Functionality $\mathcal{F}_{\textit{cmPSI}}$}}\par}
         \smallskip
        
        \textbf{Parameters:} The number of parties $n$, the size of dataset $b$, the bit length of element $l$.

        \textbf{Inputs:} For each $i \in [n]$, $P_i$ holds dataset $\textit{ID}^i$ with size $b$.

        \textbf{Functionality:}
        \begin{enumerate}[label=\arabic{enumi}.,leftmargin=1.2em,itemsep=0pt, topsep=0pt, partopsep=0pt, parsep=0pt]
            \item Compute $I_{id} =\cap_{i=1}^n(\textit{ID}^i)$.
            \item For each $j \in [b]$, if $\textit{ID}^1[j] \in I_{id}$, set $\Phi[j] = 0$, else sample $\Phi[j]$ uniformly from $\mathbb{Z}_{2^l}$.
            \item Sample $\share{\Phi}_i \leftarrow \mathbb{Z}_{2^l}^{b \times 1}$ satisfying $\Phi = \sum_{i=1}^n\share{\Phi}_i$.
            \item Return $\share{\Phi}_i$ to $P_i$ for $i \in [n]$.
        \end{enumerate}
        }
    }
    }
    \caption{Ideal functionality of circuit-based multi-party PSI}
    \label{funct:cmpsi}
\end{figure}

To illustrate the core concept, we first present our proposed cmPSI protocol in a ring-based communication structure. We then demonstrate how substantial performance gains can be achieved by optimizing the communication structure while maintaining security guarantees.

\subsubsection{cmPSI with Ring-based Communication Structure}

Holding $\textit{ID}^1$ or $\textit{ID}^i$ (with size $b$) processed by hashing, each pair $P_1$ and $P_i$ ($i \in [2, n]$) invokes an OPRF instance. Specifically, in each OPRF instance, $P_1$ is the receiver and gets $y^i_j = F_{k_i}(\textit{ID}^1[j])$ for $\forall j \in [b]$, while $P_i (i \ge 2)$ is the sender, gets the PRF key $k_i$ and computes PRF values $ U^i[j] = \{F_{k_i}(\textit{ID}^i[j][u])\}_{ \forall u \in [|\textit{ID}^i[j]|]}$ for $\forall j \in [b]$.
Subsequently, parties adopt a ring-based communication structure to iteratively transmit the OKVS tables from \(P_n\) to \(P_1\) while preserving privacy. Concretely, each party $P_i (i \ge 2)$ generates independent randomness $r^i_j \in \mathbb{Z}_{2^l}$ for $\forall j \in [b]$. The transmission begins with $P_n$, who encodes the key-value pairs $I^n = \{(\textit{ID}^n[j][u], U^n[j][u] + r_j^n)\}_{\forall j \in [b], \forall u \in [|\textit{ID}^n[j]|]}$ into an OKVS table $S^n = \textit{Encode}(I^n)$.
In our construction, each row of PRF values \(U^n[j]\) (for row \(j \in [b]\)) is masked by the same random value \(r_j^n\). Crucially, due to the OPRF interaction between \(P_1\) and \(P_n\), party \(P_1\) can learn at most one PRF output per row $j$, if and only if \(\textit{ID}^1[j]\) matches \(\mathit{id}^n_{j,u} \in \textit{ID}^n[j]\). Consequently,  values $\{U^n[j][u] + r_j^n)\}_{j \in [b],u \in [|U^n[j]|]}$ are indistinguishable from uniformly random. Furthermore, under the obliviousness guarantees of the OKVS, no other party gains additional information from $S^n = \textit{Encode}(I^n)$. Thus \( P_n \) can sends \( S^n \) to \( P_{n-1} \) with privacy preservation.
After receiving $S^n$, $P_{n-1}$ decodes $S^n$ to retrieve $V^{n-1}[j] = \{ v^i_{j,u} = \textit{Decode}(S^n, \textit{ID}^{n-1}[j][u]) \}_{\forall u \in [|\textit{ID}^{n-1}[j]|]}$ for $\forall j \in [b]$, then constructs new key-value pairs $I^{n-1} = \{(\textit{ID}^{n-1}[j][u], U^{n-1}[j][u] + V^{n-1}[j][u] + r_j^{n-1})\}_{\forall j \in [b], \forall u \in [|\textit{ID}^{n-1}[j]|]}$. 
$P^{n-1}$ then encodes new key-value pairs into the OKVS table $S^{n-1} = \textit{Encode}(I^{n-1})$ and sends $S^{n-1}$ to the next party $P_{n-2}$.
By the same obliviousness property, other parties cannot infer additional information from $S^{n-1}$.
This process is repeated iteratively from $P_{n-1}$ to $P_2$, with each $P_i$ ($i \ge 2$) receiving an OKVS table $S^{i+1}$, decoding it to retrieve the intermediate results $V^i$, constructing new key-value pairs $I^i = \{(\textit{ID}^i[j][u], U^i[j][u] + V^i[j][u] + r_j^i)\}_{\forall j \in [b], \forall u \in [|\textit{ID}^i[j]|]}$, encoding the OKVS table $S^i = \textit{Encode}(I^i)$ and sending $S^i$ to the next party $P_{i-1}$. 
Finally, upon receiving $S^2$ from $P_2$, $P_1$ decodes $S^2$ to retrieve $V^1 = \{v^1_j = \textit{Decode}(S^2, \textit{ID}^1[j])\}_{\forall j \in [b]}$ and computes its output $\share{\Phi_{b\times 1}}_1$ by subtracting $V^1$ from the sum of all received $y^i_j$ values, that is $\share{\Phi}_1=(\sum_{i=2}^{n}{y_1^i}-V^1[1],\ldots,\sum_{i=2}^{n}{y_b^i}-V^1[b])$. Each other party $P_i (\forall i \ge 2)$ sets its output $\share{\Phi}_i = (r_1^i, r_2^i, \ldots, r_b^i)$.

The key security points of the cmPSI protocol are summarized as follows: 
(1) using independent random masks per bin for secret sharing: for each $j$ ($j \in [b]$), an independent random mask \(r^i_j\) is introduced into the key-value pairs \(I^i\) used by \(P_i\) (\(i \ge 2\)) when encoding the OKVS table. This ensures that each party obtains secret-shared flags indicating intersection elements without exposing the underlying IDs;
(2) storing transformed IDs ${id\|j}$ in $\textit{ID}^i$ and using OPRF on $\textit{ID}^i$ to prevent collusion-based brute force attacks: adding the PRF value of $id\|j$ within the values of key-value pairs $I^i$, and $P_1$ invokes OPRF with each other party $P_i$ ($i \ge 2$) to remove the PRF values of intersection IDs, preventing colluding parties from inferring other parties' ID values. For example, if IDs store the original IDs instead of \( id \| j \), consider parties \(P_1\), \(P_2\), and \(P_3\), where \(P_1\) holds IDs \(a, e\), \(P_2\) holds ID \(a\), and \(P_3\) holds IDs \(a, e\), with IDs \(a\) and \(e\) hashed into the same bin for \(P_3\).
If $P_1$ and $P_2$ collude, they can infer whether $P_3$ has $e$ by verifying that $\textit{Decode}(S^3, a) - F_{k_3}(a)$ equals $\textit{Decode}(S^3, e) - F_{k_3}(e)$. However, by using OPRF on the concatenated value $id\|j$, the above attack is effectively prevented, since each hash table bin of $P_1$ has only one element, $P_1$ obtains $F_{k_3}(a \| j)$ and $F_{k_3}(e \| j')$ for $j' \neq j$ but does not obtain $F_{k_3}(e\|j)$.

\begin{small}
\begin{protocol}{$\Pi_{\rm{cmPSI }}$}{cmpsi}
\smallskip
\textbf{Parameters:} The number of parties $n$, the size of hash table $b$, the bit  length of element $l$, the security parameter $\kappa$ and the statistical security parameter $\lambda$. An OKVS scheme $\textit{Encode}(\cdot)$, $\textit{Decode}(\cdot,\cdot)$. An OPRF functionality $\mathcal{F}_{\textit{oprf}}$. A PRF $F_{(\cdot)}(\cdot)$.\\
\textbf{Inputs:} $P_i$ holds $\textit{ID}^i$ for $i \in [n]$.\\
\textbf{Outputs:} $P_i$ obtains $\share{\Phi}_i$, such that if 
$ \textit{ID}^1[j] \in  I_{id} (\forall j \in [b])$,
$\Phi[j]$ is zero; otherwise $\Phi[j]$ is randomness.\\
\textbf{Offline:} $P_i (i \ge 2)$ samples uniformly random and independent values $(r_1^i,r_2^i,\dots,r_b^i) \in \mathbb{Z}_{2^l}$. 

\textbf{Online:}
\begin{enumerate}[label=\arabic{enumi}.,leftmargin=1.2em,itemsep=0pt, topsep=0pt, partopsep=0pt, parsep=0pt]
    \item \textbf{Performing OPRF:}\
    \begin{enumerate}[label=\arabic{enumii}),leftmargin=1.2em,itemsep=0pt, topsep=0pt, partopsep=0pt, parsep=0pt]
    \item For $i \in [2,n]$, each $(P_1, P_i)$ pair invokes an instance of $\mathcal{F}_{\textit{OPRF}}$, where:\
        \begin{itemize}[leftmargin=1em,itemsep=0pt, topsep=0pt, partopsep=0pt, parsep=0pt]
            \item[-] $P_i (i \ge 2)$ is the sender and gets PRF key $k_i$.
            \item[-] $P_1$ is the receiver, inputs $\textit{ID}^1$ and receives $Y^i = (y_1^i,\cdots,y_b^i)$ where $y_j^i = F_{k_i}(\textit{ID}^1[j])$ for $\forall j \in [b]$.
        \end{itemize}    
    \item For $i \in [2,n]$, each $P_i$ computes $U^i[j] = \{ F_{k_i}(\textit{ID}^i[j][u])\}_{\forall u \in [|\textit{ID}^i[j]|]}$ for $\forall j \in [b]$.
    \end{enumerate}
    
    \item \textbf{Transmiting OKVS tables:} 
\begin{enumerate}[label=\arabic{enumii}),leftmargin=1.2em,itemsep=0pt, topsep=0pt, partopsep=0pt, parsep=0pt]
    \item Each $P_i (i \in [n])$ uses the optimized communication structure generation algorithm\ref{alg:opt} to obtain the set of child parties $P_i.children$ and its parent party $P_i.parent$  (the child party sends the OKVS table to the parent party).
    \item  For $i = n$ to $2$: 
    \begin{enumerate}[leftmargin=1em,itemsep=0pt, topsep=0pt, partopsep=0pt, parsep=0pt]
        \item If $P_i.children \neq \emptyset$, $P_i$ receives $\{S^{i^*}\}_{\forall  i^* \in P_i.children } $ from all child parties and gets $ V^{i,i^*} =\{ v^{i,i^*}_{j,u} = \textit{Decode}(S^{i^*},\textit{ID}^i[j][u])\}_{\forall j \in [b], \forall u \in [|\textit{ID}^i[j]|]} $ for $\forall i^* \in P_i.children$. 
        \item If $i \ge 2$, $P_i$ sets $I^i = \{(\textit{ID}^i[j][u], \sum_{{i^*} \in P_i.children} V^{i,i^*}[j][u] + U^i[j][u] + r^i_j )\}_{\forall j \in [b], \forall u \in [|\textit{ID}^i[j]|]}$.  $P_i$ builds OKVS table $S^i = \textit{Encode}(I^i)$ and sends $S^i$ to $P_i.parent$.
    \end{enumerate}
\end{enumerate}
    \item \textbf{Setting outputs:}\
    \begin{itemize}[leftmargin=0em,itemsep=0pt, topsep=0pt, partopsep=0pt, parsep=0pt]
        \item[-] $P_1$ sets $\share{\Phi}_1 = (\sum_{i=2}^{n} y_1^i - V^1[1], \dots, \sum_{i=2}^{n} y_b^i - V^1[b])$.
        \item[-] $P_i (i \ge 2)$ sets $\share{\Phi}_i = (r_1^i, r_2^i, \dots, r_b^i)$.
    \end{itemize}
\end{enumerate}

\end{protocol}
\end{small}

\subsubsection{cmPSI with Optimized Communication Structure} Based on our earlier security analysis, an adversary cannot reconstruct the original IDs from an encoded OKVS table via brute-force attacks. Therefore, it is secure to transmit OKVS tables to any party, enabling efficient parallel transmissions that ultimately converge at $P_1$.

As is shown in Figure~\ref{fig:cmpsi-structure}, a scenario for cmPSI involving 5 parties.
The leftmost diagram in Figure~\ref{fig:cmpsi-structure} illustrates a ring-based communication structure for transmitting OKVS tables in our proposed cmPSI protocol. 
We denote the transmission time for an OKVS table (excluding latency) as $t_s$ and the network latency as $t_l$.
The total time required to transmit all the OKVS tables via a ring-based structure is $4(t_s + t_l)$, assuming that the local encoding and decoding times are negligible. 
The middle part of Figure~\ref{fig:cmpsi-structure} illustrates the simplest star-based parallel structure, in which every party $P_i(i \ge 2)$ sends its OKVS table $S^i$ to $P_1$. The total time required to transmit all the OKVS tables via the star-based structure is $4t_s + t_l$, which is $(n-2)t_l$ less than that of the ring-based structure.
We notice that if $P_1$ first receives $P_4$'s OKVS table, then parties $P_2$,$P_3$, and $P_5$ remain idle. As is shown in the rightmost part of Figure~\ref{fig:cmpsi-structure}, this idle time can be effectively leveraged by having $P_5$ send its OKVS table to $P_2$. $P_2$ subsequently decodes the OKVS table $S^5$ to retrieve $V^{2,5}[j] = \{v^{2,5}_{j,u} = \textit{Decode}(S^5, U^2[j][u])\}_{\forall u \in [|U^2[j]|]}$ for $\forall j \in [b]$, and add the retrieved values into values of its own key-value pairs, constructing new pairs $I^{2} =\{ (\textit{ID}^2[j][u], U^2[j][u] + V^{2,5}[j][u] + r^2_j)\}_{\forall j \in [b], \forall u \in [|\textit{ID}^2[j]|]}$. $P_2$ encodes the OKVS table $S^{2} = \textit{Encode}(I^{2})$ and sends $S^2$ to party $P_1$. 
As all parties possess datasets of equal size, $P_1$ and $P_2$ simultaneously receive the OKVS table from $P_4$ and $P_5$, respectively. Subsequently, $P_1$ receives the OKVS tables $S^3$ and $S^2$ from $P_3$ and $P_2$, respectively. The time required to transmit all the OKVS tables via this optimized structure is $3t_s+t_l$, which is $t_s$ less than that of the star-based structure.

\begin{figure}[htbp]
    \centering
    \captionsetup[subfigure]{labelfont={scriptsize}, textfont={scriptsize}}
    \subfloat[\scriptsize Ring-based structure\label{fig:ring-comm}]{%
        \includegraphics[width=0.14\textwidth]{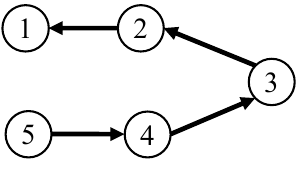}%
    }
    \hspace{1.6em}
    \subfloat[\scriptsize Star-based structure\label{fig:star-comm}]{%
        \includegraphics[width=0.14\textwidth]{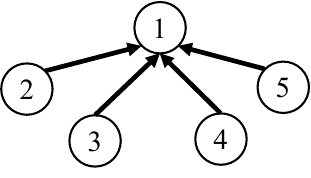}%
    }
    \hspace{1.6em}
    \subfloat[\scriptsize Optimized structure\label{fig:opt-comm}]{%
        \includegraphics[width=0.13\textwidth]{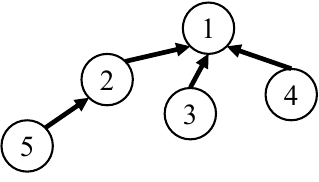}%
    }
    \caption{Comparison of different communication structures for transmitting OKVS tables. Each circle denotes a party and the number $i$ in the circle denotes $P_i$ ($i \in [n]$). Arrow indicates the transmission direction of the OKVS table. Each party must finish receiving the OKVS tables before sending. 
    }
    \label{fig:cmpsi-structure}
\end{figure}

We propose an optimized communication structure generation algorithm (detailed in Algorithm~\ref{alg:optstruct} in the Appendix) that takes as input the number of parties $n$, the number of designated root party $leader\_num$ ($leader\_num$ is 1 in our proposed cmPSI protocol), the time required to send one OKVS table $t_s$, and the network delay $t_l$. It outputs the root node of the optimized communication structure, where each node represents a party and contains its identifier, parent, and children. Each party must finish receiving all the OKVS tables from its child party before it sends the OKVS table to the parent party. The algorithm follows a greedy approach to maximize communication utilization while minimizing the latency for each party.
As is shown in Protocol~\ref{pro:cmpsi}, we construct the final cmPSI protocol by using this optimized communication structure to efficiently transmit the OKVS tables in step 2.

\begin{theorem}
The protocol $\Pi_{\rm {cmPSI}}$ (Protocol ~\ref{pro:cmpsi}) securely realizes  $\mathcal{F}_{\textit{cmPSI}}$ (Figure \ref{funct:cmpsi}) in the random oracle and $\mathcal{F}_{\textit{OPRF}}$-hybrid model, against semi-honest adversary $\mathcal{A}$ who can corrupt $n-1$ parties.
\end{theorem}

\begin{proof} 
We exhibit simulators $\textit{Sim}_C$ for simulating the view of the corrupt parties including $P_1$ and excluding $P_1$, respectively, and prove that the simulated view is indistinguishable from the real one via standard hybrid arguments. Let $C$ be the set of corrupted $n-1$ parties. 

\noindent \underline{\textbf{Case 1 ($P_1 \in C$).}} 
The $\textit{Sim}_C$ samples randomness $y'_j \leftarrow \mathbb{Z}_{2^l}$ for $\forall j \in [b]$ and invokes the OPRF receiver’s simulator $\textit{Sim}_{\textit{OPRF}}^R(Y^i,Y'^i)$, where $Y'^i= \{ y_1',\ldots,y_b' \} $ and sends $Y'^i$ to $P_1$ on behalf of $P_i (\notin C)$. 
If $P_i (\notin C$) has child parties, $\textit{Sim}_C$ receives the OKVS table from $P_i$'s child parties. 
Next, $\textit{Sim}_C$ samples uniformly random OKVS table $S'^{i} \leftarrow \mathbb{Z}_{2^l}^{m'\times 1}$ where $m'$ is the size of the OKVS table when encoding $h\cdot m$ elements. And $\textit{Sim}_C$ sends $S'^i$ to $P_i (\notin C)$'s parent party.
We argue that the view output by $\textit{Sim}_C$ is indistinguishable from the real one. First, we define three hybrid transcripts $T_0$,$T_1$,$T_2$, where $T_0$ is the real view of $C$, and $T_2$ is the output of $\textit{Sim}_C$.
\begin{enumerate}[label=\arabic{enumi}.,leftmargin=1em,itemsep=0pt, topsep=0pt, partopsep=0pt, parsep=0pt]
    \item $\textit{Hybrid}_0$.The first hybrid is the real interaction in Protocol~\ref{pro:cmpsi}. Here, an honest party $P_i$ uses real inputs and interacts with corrupt parties $C$. Let $T_0$ denote the real view of $C$.
    \item $\textit{Hybrid}_1$. Let $T_1$ be the same as $T_0$, except that the OPRF execution is replaced by the OPRF receiver’s simulator. $\textit{Sim}_{\textit{OPRF}}^R(Y^i, Y'^i)$, where $Y'^i$ has $b$ elements $\{ y_1',\ldots,y_b' \}$. The simulator security of OPRF guarantees that this view is indistinguishable from $T_0$.
    \item $\textit{Hybrid}_2$. Let $T_2$ be the same as $T_1$, except that the OKVS table $S'^i$ for $P_i (\notin C)$ is sampled uniformly from $\leftarrow \mathbb{Z}_{2^l}^{m' \times 1}$. 
    In real execution, $P_i (\notin C)$ sets $I^i = \{(\textit{ID}^i[j][u], \sum_{P_{i^*} \in P_i.children} V^{i^*}[j][u] + U^i[j][u] + r^i_j )\}_{\forall j \in [b], \forall u \in [|\textit{ID}^i[j]|]}$ and builds the OKVS table $S^i = \textit{Encode}(I^i)$.  If $x = \textit{ID}^i[j][u]$ is not the intersection element of $P_1$ and $P_i (\notin C)$, the value $I^i(x)$ is uniformly random  from $\mathbb{Z}_{2^l}$ because it contains the PRF value $U^i[j][u] = F_{k_i}(x)$. If $x$ is the intersection element of $P_1$ and $P_i$, $P_1$ holds $U^i[j][u] = F_{k_i}(x)$, but $r_j^i$ is a uniform sample by $P_i$, so the values $I^i(x)$ are still uniformly at random from $\mathbb{Z}_{2^l}$. Owing to the double obliviousness and random decoding property of the OKVS, the simulated $S'^i$ has the same distribution as $S$ in the real protocol. Hence, $T_1$ and $T_2$ are statistically indistinguishable. This hybrid is exactly the view output by the simulator.
\end{enumerate}

\noindent \underline{\textbf{Case 2 ($P_1 \notin C$).}} 
For each $ P_i (i \ge 2)$, $\textit{Sim}_C$ generates random key $k_i$. It then invokes the OPRF sender’s simulator $\textit{Sim}_{\textit{OPRF}}^S(k_i)$ and appends the output to the view. $\textit{Sim}_C$ receives the OKVS table from $P_1$'s children parties on behalf of $P_1$.
We argue that the view output by $\textit{Sim}_C$ is indistinguishable from the real one. We first define two hybrid transcripts $T_0$,$T_1$, where $T_0$ is the real view of $C$, and $T_1$ is the output of $\textit{Sim}_C$.
\begin{enumerate}[label=\arabic{enumi}.,leftmargin=1em,itemsep=0pt, topsep=0pt, partopsep=0pt, parsep=0pt]
    \item $\textit{Hybrid}_0$.The first hybrid is the real interaction in Protocol~\ref{pro:cmpsi}. Here, an honest party $P_i$ uses real inputs and interacts with the corrupt parties $C$. Let $T_0$ denote the real view of $C$.
    \item $\textit{Hybrid}_1$. Let $T_1$ be the same as $T_0$, except that the OPPF execution is replaced by the OPRF sender’s simulator $\textit{Sim}_{\textit{OPRF}}^S(k_i)$ for every $i \in [2,n]$. The simulator security of OPRF guarantees this view is indistinguishable from $T_0$. This hybrid is exactly the view output by the $\textit{Sim}_C$.
\end{enumerate}

\end{proof}

\subsection{Secure Feature Alignment}

\subsubsection{Definition}
We formally define secure feature alignment functionality $\mathcal{F}_{\textit{smFA}}$ in Figure~\ref{funct:sfm}. Using the secure feature alignment protocol, each party $P_i (i \in [n])$ holds $(\share{\Phi}_i, \textit{ID}^i, F^i)$ and obtains secret-shared and joint training dataset $\share{\mathcal{D}}_i$ consisting of aligned features.

\begin{figure}[htbp]
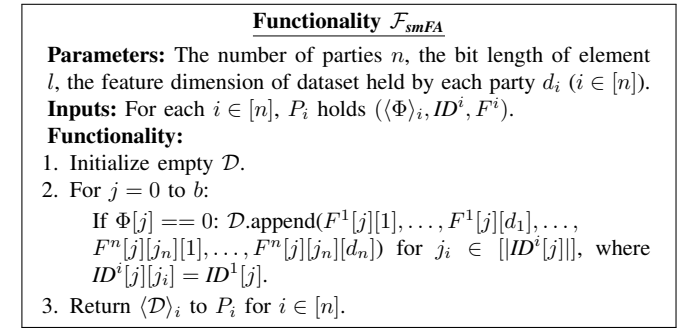

    \centering
    \fbox{
    \scalebox{0.85}{
        \parbox{0.52\textwidth}{
        {\centering \textbf{\underline{Functionality $\mathcal{F}_{\textit{smFA}}$}}\par}
        \smallskip
        \textbf{Parameters:} The number of parties $n$, the bit length of element $l$, the feature dimension of dataset held by each party $d_i$ ($i \in [n]$).

        \textbf{Inputs:} For each $i \in [n]$, $P_i$ holds $(\share{\Phi}_i,\textit{ID}^i, F^i)$.

        \textbf{Functionality:} 
        \begin{enumerate}[label=\arabic{enumi}.,leftmargin=1em,itemsep=0pt, topsep=0pt, partopsep=0pt, parsep=0pt]
        \item Initialize empty $\mathcal{D}$. 
        \item For $j=0$ to $b$:
            \begin{itemize}[label={}]
            \item If $\Phi[j] == 0$: $\mathcal{D}$.append($F^1[j][1], \ldots, F^1[j][d_1], \ldots , \\ F^n[j][j_n][1] , \ldots, F^n[j][j_n][d_n]$) for $j_i \in [|\textit{ID}^i[j]|]$, where $\textit{ID}^i[j][j_i] = \textit{ID}^1[j]$.
            \end{itemize}
        \item Return $\share{\mathcal{D}}_i$ to $P_i$ for $i \in [n]$.
        \end{enumerate}

        }
    }
    }
    \caption{Ideal functionality of secure multi-party feature alignment}
    \label{funct:sfm}
\end{figure}

\subsubsection{Construction}
Protecting intersection IDs in \SecVPre introduces a significant challenge because without revealing which IDs are in the intersection, parties cannot directly determine which features need to be aligned or how they should be ordered. To address this issue, the smFA protocol has two milestones. First, parties use OKVS and secret sharing to obtain a secret-shared dataset consisting of aligned features and redundant data. Second, parties use a secure multi-party shuffle to remove redundant data corresponding to non-intersection IDs, thus yielding a secret-shared and joint dataset consisting of aligned features.
The smFA protocol $\Pi_{\rm smFA}$ is described formally in Protocol~\ref{pro:fa}.

The first milestone corresponds to steps 1–3 of \(\Pi_{\rm {smFA}}\) (Protocol~\ref{pro:fa}). In the offline phase, each party \( P_i \) (\( i \geq 2 \)) samples a matrix of uniformly random and independent values \( R^i_{b \times d_i} = [r_{j,k}^i]_{\forall j \in [b], \forall k \in [d_i], r_{j,k}^i \in \mathbb{Z}_{2^l} }\).
In the online phase,  $P_1$ secret shares its feature matrix $F^1$ with all other parties $P_i$ ($i \ge 2$) (step 1). Next,
each \( P_i \) (\( i \in [2,n] \)) secret shares \( R^i \) among all parties except $P_1$, denoted as \( \share{R^i} = \mathit{Shr}(R^i, i, [n] \setminus \{1\}) \) (step 2-(1)). Subsequently, each party \(P_i\)(\(i\ge2\)) constructs an OKVS table \(S_f^i = \textit{Encode}(I^i_f)\), where \(\ I^i_f = \{ (U^i[j][u],H_o(U^i[j][u]) \oplus ((F^i[j][u][1] - r^i_{j,1})\|\ldots\|(F^i[j][u][d_i] - r^i_{j,d_i})) )\}_{\forall j\in[b], \forall u \in [|\textit{ID}^i[j]|]}\) (step 2-(2)). Here, \(U^i=F_{k_i}(\textit{ID}^i)\) is $P_i$'s PRF value (obtained from \(\Pi_{\rm {cmPSI}}\)), and $H_o:\{0,1\}^\kappa \rightarrow \{0,1\}^{*}$ is public random oracles whose output bit length `$*$' is $l\cdot d_i$. Each $P_i (i \ge 2)$ sends the OKVS table $S_f^i$ to \(P_1\) (step 2-(2)). After receiving $S_f^i$, $P_1$ encodes $S_f^i$ using the OPRF values $Y^i$ (obtained from $\Pi_{\rm {cmPSI}}$) to retrieve $V^i_f = \{ v^i_j = \textit{Decode}(S^i_f,Y^i[j])\}_{\forall j \in [b]}$ (step 2-(3)). \(P_1\) XORs these result with \(H_o(Y^i)\) to derive masked features $V^i_f = V^i_f \oplus H_o(Y^i)$. 
In step 3, \( P_1 \) sets \(\share{Z}_1 = (\share{F^1}_1, V^2_f, V^3_f, \ldots, V^n_f)\), while each \( P_i \) (\( i > 2 \)) sets \( \share{Z}_i = (\share{F^1}_i, \share{R^2}_i, \ldots, \share{R^n}_i)\).
Through these steps, the parties collectively obtain a secret-shared dataset \(\share{Z}\) that includes aligned features and redundant data.

\begin{figure}[htbp]
    \centering
    \includegraphics[width=0.48\textwidth]{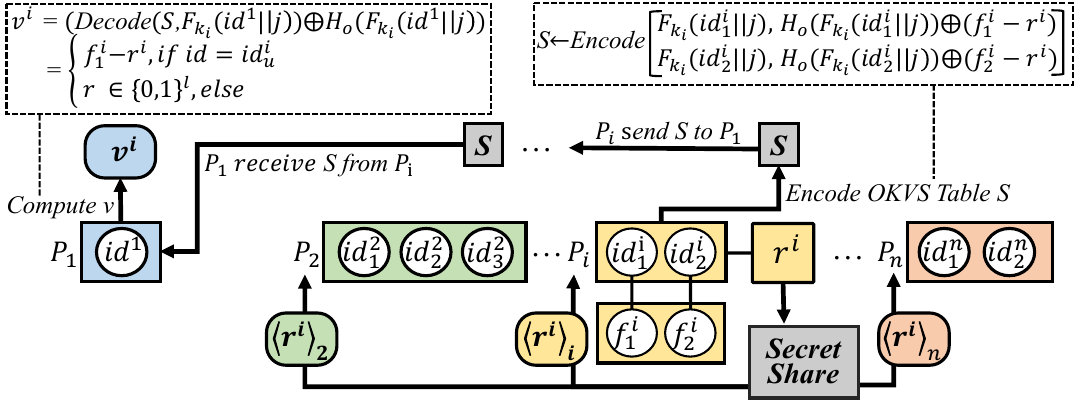}
    \caption{Examples of sharing features of party $P_i (i \ge 2)$ in $j$-th bin (index $j$ omitted for simplicity) with privacy-preserving. $P_i$'s bin contains two ids $id^i_1,id^i_2$ and the corresponding features are $f^i_1,f^i_2$. $P_i$ samples a random value $r^i$ locally and secretly shares $r^i$ to $P_j (\forall j \in [2,n])$, who each obtain share $\share{r^i}_j$. 
    $P_i$ sends the OKVS table $S$ to $P_1$, who decodes $S$ and computes $v^i$. The dashed boxes highlight computations corresponding directly to the adjacent text, and ellipses indicate outputs held by each party.}
    \label{fig:fa_instance}
\end{figure}

To analyze the correctness of the above steps, consider the \(j\)-th bin in Figure \ref{fig:fa_instance}. For simplicity, we omit the index \( j \). Party \( P_i \)'s dataset has a single feature dimension, containing two ID values $id^i_1, id^i_2$ in the \( j \)-th bin. 
If $id^1 = id^i_u$ for some $u \in [|\textit{ID}^i[j]|]$, $P_1$ computes
\begin{align}
v^i &= (\textit{Decode}(S,F_{k_i}(id^1\|j)) \oplus (H_o(F_{k_i}(id^1)) \nonumber
\\ &= (\textit{Decode}(S,F_{k_i}(id^i_u\|j)) \oplus (H_o(F_{k_i}(id^i_u)) \nonumber
\\ &= H_o(F_{k_i}(id^i_u))\oplus(f^i_u - r^i) \oplus (H_o(F_{k_i}(id^i_u)) \nonumber
\\ &= f^i_u - r^i
\end{align}
Adding the shares \(\share{r^i}\)(held by \(P_2 ,\ldots, P_n\)) to \(v^i\) (held by \(P_1\)) reconstructs \(f^i_u\), meaning that all parties hold a share $\share{z}$\footnote{Here, $v^i$ is $V^i_f[j]$, $z$ is $Z[j][\sum_{i'=1}^{i-1}d_{i'}]$, $r^i$ is $r^i_j$, $id^i$ is $id^i_j$, as we omit $j$ for simplicity.} of \(P_i\)’s feature $f^i_u$ for intersection ID.
If \(id^1 \neq id^i_u\), the decoded output is a random value. Adding shares \(\share{r^i}\) (held by \(P_2, \ldots, P_n\)) to \(v^i\) reconstructs a random value, implying that each party obtains a secret-shared random value $\share{z}$. Thus, all parties obtain $\share{Z_{b \times d}}_i$, where if $\textit{ID}^1[j] \in I$, the $j$-th row of $\share{Z}$ consists of secret-shared aligned features from all parties corresponding to this ID; otherwise, redundant data.

Regarding security, note that for the key-value pairs \(I^i_f\) encoded into $P_i$'s OKVS table $S^i_f$ ($i \ge 2$), party $P_i$ set \( U^i = F_{k_i}(id^i_u) \) as the keys of \( I^i_f \). And $P_i$ mask each feature value using a random number and \( H_o(U^i) \), thus forming the values $\{H_o(U^i[j][u]) \oplus ((F^i[j][u][1] - r^i_{j,1})\|\ldots\|(F^i[j][u][d_i] - r^i_{j,d_i}))\}_{\forall j\in[b], \forall u \in [|\textit{ID}^i[j]|]}$ of \( I^i_f \).  
This design prevents \( P_1 \) from inferring whether any decoded entry from $S^i_f$ is a random value or a masked feature. Moreover, since $P_i$ uses OKVS to encode the entire dataset $F^i$ across the simple hash table, $S^i_f$ does not reveal the number of elements in each hash bin.

For the second milestone for removing redundant data in step 4 of $\Pi_{\rm smFA}$, parties employ our proposed secure shuffling protocol (Protocol \ref{pro:secshuffle}) on  $\share{Z}\|\share{\Phi}$, reconstruct shuffled flags $\share{\hat{\Phi}}$ into plaintext $\hat{\Phi}$, and then remove redundant data from $\share{Z}$ if the corresponding entries in $\hat{\Phi}$ are randomness.

\begin{small}
\begin{protocol}{$\Pi_{\rm{smShuffle}}$}{secshuffle}

        \smallskip

        \textbf{Inputs:} $P_i$ holds a secret-shared matrix $\share{X}_i$ with size $b \times d$.

        \textbf{Outputs:} $P_i$ obtains a shuffled secret-shared matrix $\share{X'}_i$, such that ${X'} = \pi({X})$. where $\pi: [b] \to [b]$ is an unknown permutation to all parties.


        \textbf{Offline}: 
        $P_i (\forall i \in [n])$ samples uniformly random independent values $R^{i,i'} = [ r^{i,i'}_{j,u}]_{ \forall j \in [b], \forall u \in [d], r^{i,i'}_{j,u} \in \mathbb{Z}_{2^l}}$ for $\forall i' \in [n] \setminus \{i\}$) and a random permutation $\pi^{i}:[b] \rightarrow [b]$. Each pair of $P_i$ and $P_{i'}$($\forall i, i' \in [n], i \neq i'$) invokes $\Pi_{\rm {O-Shuffle}}$:
            \begin{itemize}[leftmargin=1em,itemsep=0pt, topsep=0pt, partopsep=0pt, parsep=0pt]
            \item[-] $P_{i}$ inputs $\pi^i$ and obtains $\share{\hat{R}^{i',i}}_{i}$.
            \item[-] $P_{i'}$ inputs $R^{i',i}$ and obtains $\share{\hat{R}^{i',i}}_{i'}$.
            \end{itemize}

        \textbf{Online}: 
        \begin{enumerate}[label=\arabic{enumi}.,leftmargin=1.2em,itemsep=0pt, topsep=0pt, partopsep=0pt, parsep=0pt]
        \item $P_i (\forall i \in [n])$ sets $\share{\hat{X}}_i$ = $\share{X}_i$.
        \item For $i = 1$ to $n$: 

        \begin{enumerate}[label=\arabic{enumii}),leftmargin=1.2em,itemsep=0pt, topsep=0pt, partopsep=0pt, parsep=0pt]
            \item $P_{i'}$ ($\forall i' \in [n] \setminus \{i\}$) computes $W^{i',i} = \share{\hat{X}}_{i'} -  R^{i',i}$ and sends $W^{i',i}$ to $P_i$. $P_{i'}$ updates $\share{\hat{X}}_{i'} = \share{\hat{R}^{i',i}}_{i'}$.
            \item Upon receiving all $W^{i',i}$ from $P_{i'}$ ($i' \in [n] \setminus \{i\}$), $P_i$ computes $W^i = \sum_{i' \in [n] \setminus \{i\}} W^{i',i} + \share{\hat{X}}_i$. 
            Then, $P _i$ applies its permutation $\pi^i$ to obtain $\hat{W}^i = \pi^i(W^i)$ and updates $\share{\hat{X}}_i = \hat{W}^ {i} + \sum_{i' \in [n] \setminus \{i\}}\share{\hat{R}^{i',i}}_i$.
        \end{enumerate}
        \end{enumerate}
\end{protocol}
\end{small}

In our proposed secure multi-party shuffle $\Pi_{\rm smShuffle}$ (Protocol \ref{pro:secshuffle}), each party $P_i (i \in [n])$ samples a private random permutation $\pi^i$ and generates random masks $R^{i,i'} (i' \in [n]\setminus \{i\})$ for each other party in the offline phase. Then, through a two-party protocol $\Pi_{\rm {O-Shuffle}}$, parties convert these masks into secret-shared, permuted forms $\share{\hat{R}^{i,i'}} (i' \in [n]\setminus \{i\})$. In the online phase, the parties iteratively update their shares in $n$ rounds. In each round, a designated party $P_i$ ($i \in [n]$) collects masked inputs $W^{i',i} = \share{\hat{X}}_{i'} -  R^{i',i}$ from all others $P_{i'} (i' \in [n] \setminus \{i\})$, applies its private permutation $\pi^i$ to the aggregated value $\hat{W}^i = \pi^i(W^i) = \pi^i( \sum_{i' \in [n] \setminus \{i\}} W^{i',i} + \share{\hat{X}}_i)$, and updates its share $\share{\hat{X}}_i = \hat{W}^{i} + \sum_{i' \in [n] \setminus \{i\}}\share{\hat{R}^{i',i}}_i$, while the others $P_{i'}$ update their shares $\share{\hat{X}}_{i'} = \share{\hat{R}^{i',i}}_{i'}$. After all rounds, the joint output is a secret-shared matrix, which is a randomly shuffled version of the original input, with the shuffle pattern and data remaining completely hidden from all the parties.

Since the secure shuffle disrupts the original order of flags $\hat{\Phi}$ and each plaintext entry in $\hat{\Phi}$ is either 0 or randomness, ensuring that no party can infer the original intersection IDs from $\hat{\Phi}$. Then, the parties can remove redundant data:
if \(\hat{\Phi}[j] = 0\), it indicates that the corresponding row $\share{Z[j]}$ belongs to an intersection ID, and thus \( \share{Z[j]} \) is appended in final secret-shared dataset $\share{\mathcal{D}}$.
Finally, each party $P_i$ obtains the secret-shared and joint dataset $\share{\mathcal{D}_{c \times d}}_i$ consisting of secret-shared and aligned features corresponding to the intersection IDs $I_{id}$.

\begin{small}
\begin{protocol}{$\Pi_{\rm{smFA}}$}{fa}
\smallskip
\textbf{Parameters:} The number of parties $n$, the size of input dataset $b$, the bit length of element $l$, and the feature dimension $d_i$ for party $P_i$($i \in [n]$), the security parameter $\kappa$ and the statistical security parameter $\lambda$. An OKVS scheme $\textit{Encode}(\cdot)$, $\textit{Decode}(\cdot,\cdot)$. Random oracles $H_o:\{0,1\}^\kappa \rightarrow \{0,1\}^{*}$.
\\
\textbf{Inputs:} $P_i$ holds $(\share{\Phi}_i, \textit{ID}^i, F^i)$ for $i \in [n]$. $P_1$ holds OPRF values $\{Y^i=F_{k_i}(\textit{ID}^1)\}_{\forall i \in [2,n]}$ and $P_i (i \ge 2)$ holds PRF values $U^i = F_{k_i}(\textit{ID}^i)$ from $\Pi_{\rm {cmPSI}}$.\\
\textbf{Outputs:} $P_i$ obtains $\share{\mathcal{D}_{c \times d}}_i$ consisting of aligned features.\\
\textbf{Offline:}\
    \begin{itemize}[leftmargin=1em,itemsep=0pt, topsep=0pt, partopsep=0pt, parsep=0pt]
        \item[-] Each $P_i (i \ge 2)$ uniformly and independently samples each element of matrix $R^i_{b \times d_i} = [r_{j,k}^i]_{\forall j \in [b], \forall k \in [d_i], r_{j,k}^i \in \mathbb{Z}_{2^l}}$.
        \item[-] $P_i (\forall i \in [n])$ executes offline operation in $\Pi_{\rm {smShuffle}}$. 
    \end{itemize}
\textbf{Online:}\
\begin{enumerate}[label=\arabic{enumi}.,leftmargin=1.2em,itemsep=0pt, topsep=0pt, partopsep=0pt, parsep=0pt]
    \item \textbf{Sharing features of $P_1$:} $\share{F^1} = \mathit{Shr}(F^1, 1, [n])$.
    \item \textbf{Sharing features of $P_i$($i \ge 2$):}\
    \begin{enumerate}[label=\arabic{enumii}),leftmargin=1.2em,itemsep=0pt, topsep=0pt, partopsep=0pt, parsep=0pt]
    \item For $i \in [2,n]$, $\share{R^i} = \mathit{Shr}(R^i, i, [n] \setminus \{1\})$.
    \item $P_i (i \ge 2)$ builds OKVS table $S^i_f = \textit{Encode}(I^i_f)$, where $I^i_f = \{ (U^i[j][u],H_o(U^i[j][u]) \oplus ((F^i[j][u][1] - r^i_{j,1})\|\ldots\|(F^i[j][u][d_i] - r^i_{j,d_1})) )\}_{\forall j\in[b], \forall u \in [|U^i[j]|]}$ and sends $S^i_f$ to $P_1$.
    \item After receiving $S^i_f$ from $P_i (i \ge 2)$, $P_1$ gets $V^i_f = \{v^i_{j} = \textit{Decode}(S^i_f,Y^i[j])\}_{\forall j \in [b]}$ and computes $V^i_f = V^i_f \oplus H_o(Y^i)$ for $\forall i \in [2,n]$.
    \end{enumerate}
    \item \textbf{Merging features:}\
    \begin{itemize}[leftmargin=1em,itemsep=0pt, topsep=0pt, partopsep=0pt, parsep=0pt]
        \item[-] $P_1$ sets $\share{Z}_1 = (\share{F^1}_1, V^2_f, V^3_f, \ldots, V^n_f)$.
        \item[-] For $i \ge 2$, $P_i$ sets $\share{Z}_i = (\share{F^1}_i, \share{R^2}_i, \ldots, \share{R^n}_i)$.
    \end{itemize}
    \item \textbf{Removing redundant data:}\
    \begin{enumerate}[label=\arabic{enumii}),leftmargin=1.2em,itemsep=0pt, topsep=0pt, partopsep=0pt, parsep=0pt]
    \item $\share{\hat{Z}} \| \share{\hat{\Phi}} = \Pi_{\rm {smShuffle}}(\share{Z} \| \share{\Phi})$.
    \item $\hat{\Phi} = \mathit{Rec}(\share{\hat{\Phi}})$.
    \item For $j = 0$ to $b$: If $\hat{\Phi}[j] == 0$: $\share{\mathcal{D}}$.append$(\share{\hat{Z}[j]})$.
    \end{enumerate}
\end{enumerate}
\end{protocol}
\end{small}

\begin{theorem}
The protocol $\Pi_{\rm{smFA}}$ (Protocol~\ref{pro:fa}) is securely realized $\mathcal{F}_{\textit{smFA}}$ in random oracle model, against a semi-honest adversary $\mathcal{A}$ who can corrupt $n-1$ parties.
\end{theorem}

\begin{proof} 
We exhibit simulators $\textit{Sim}_C$ for simulating the view of the corrupt parties, including $P_1$ and excluding $P_1$, respectively, and prove that the simulated view is indistinguishable from the real one via standard hybrid arguments. Let $C$ be the set of corrupted $n-1$ parties. 

\noindent \underline{\textbf{Case 1 ($P_1 \in C$).}}
The corrupted parties obtain shares in steps 1, 2-(1), and 4, thus $\textit{Sim}_C$ can pick shares of random value on behalf of $P_i$ ($\notin C$) in these steps.
$\textit{Sim}_C$ samples uniformly random OKVS table $S'^{i} \leftarrow \mathbb{Z}_{2^l}^{m'\times d_i}$ where $m'$ is the size of OKVS table when encoding $h \cdot m$ number of elements. And $\textit{Sim}_C$ sends $S'^i$ to $P_1$ in step 2-(2).
We argue that the view output by $\textit{Sim}_C$ is indistinguishable from the real one. We first define three hybrid transcripts $T_0$,$T_1$,$T_2$, where $T_0$ is the real view of $C$, and $T_2$ is the output of $\textit{Sim}_C$.

\begin{enumerate}[label=\arabic{enumi}.,leftmargin=1em,itemsep=0pt, topsep=0pt, partopsep=0pt, parsep=0pt]
    \item $\textit{Hybrid}_0$.The first hybrid is the real interaction in Protocol~\ref{pro:fa}. Here, an honest party $P_i$ uses real inputs and interacts with the corrupt parties $C$. Let $T_0$ denote the real view of $C$.
    \item $\textit{Hybrid}_1$. Let $T_1$ be the same as $T_0$, except that the OKVS table $S'^i$ for $P_i (\notin C)$ is sampled uniformly from $\leftarrow \mathbb{Z}_{2^l}^{m' \times d_i}$. 
    By the double obliviousness and random decoding property of OKVS, the simulated $S'^i$ has the same distribution as $S$ in the real protocol. Hence, $T_1$ and $T_0$ are statistically indistinguishable.
    \item $\textit{Hybrid}_2$. Let $T_1$ be the same as $T_0$, except that the secret sharing operation is replaced by the secret share simulator. The underlying secret sharing guarantee $T_2$ is indistinguishable from $T_1$.This hybrid is exactly the view output by the simulator $\textit{Sim}_C$.
\end{enumerate}

\noindent \underline{\textbf{Case 2 ($P_1 \notin C$).}} 
The corrupted parties obtain shares in steps 1, 2-(1), and 4, thus the $\textit{Sim}_C$ can pick shares of random value to the view on behalf of $P_1$ in these steps. $\textit{Sim}_C$ receive OKVS table from $P_i$ on behalf of $P_1$ in step 2-(3).
We argue that the view output by $\textit{Sim}_C$ is indistinguishable from the real one. We first define two hybrid transcripts $T_0$, $T_2$, where $T_0$ is the real view of $C$, and $T_1$ is the output of $\textit{Sim}_C$.

\begin{enumerate}[label=\arabic{enumi}.,leftmargin=1em,itemsep=0pt, topsep=0pt, partopsep=0pt, parsep=0pt]
    \item $\textit{Hybrid}_0$.The first hybrid is the real interaction in Protocol~\ref{pro:fa}. Here, an honest party $P_i$ uses real inputs and interacts with the corrupt parties $C$. Let $T_0$ denote the real view of $C$.
    \item $\textit{Hybrid}_1$. Let $T_1$ be the same as $T_0$, except that the secret sharing operation is replaced by the secret share simulator. The underlying secret sharing guarantees that $T_1$ is indistinguishable from $T_0$.This hybrid is exactly the view output by the simulator $\textit{Sim}_C$.
\end{enumerate}

\end{proof}

\subsection{Communication Complexity}
\label{sec:design_comm}

As is shown in Table \ref{tab:comm-cmp},
we compare the communication costs of our \SecVPre with \baseline.Thee communication costs of key cryptographic primitives are defined as follows: $C_{\textit{oprf}}$, $C_{\textit{okvs}}$,$C_{\textit{opprf}}$ are $\mathcal{O}(m(\lambda + \log m))$, $\mathcal{O}(\kappa m)$ and $\mathcal{O}(m(\lambda + \log m + \kappa))$, respectively.

\begin{table}[htbp]
\belowrulesep=0pt 
\aboverulesep=0pt
\centering
\caption{Comparison of Communication sizes between our \SecVPre vs. the SOTA \baseline, where \( n \) represents the number of participating parties, \( d_i \) denotes the feature dimension of party \( P_i \), \( d \) denotes the total feature dimension, \( m \) denotes the dataset size and \( l \) denotes the bit length.}
\fontsize{7pt}{8pt}\selectfont
\setlength{\tabcolsep}{1pt}
\renewcommand{\arraystretch}{1.12}
\begin{tabular}{l|c}
\toprule
Framework & Communication Sizes \\ \midrule
\baseline & $\mathcal{O}(C_{\textit{oprf}} + (1+d_2)C_{\textit{opprf}} + (d_1 + 2d + 1)ml)$ \\
Ours ($n=2$) & $\mathcal{O}(C_{\textit{oprf}} + (1 + d_2)C_{\textit{okvs}} + (d_1 + 2d + 1)ml)$ \\
Ours & $\mathcal{O}(C_{\textit{oprf}} + (n-1 + \sum_{i=2}^n(d_i))C_{\textit{okvs}} + (d_1 + n^2d-nd + 1)ml)$ \\
\bottomrule
\end{tabular}
\label{tab:comm-cmp}
\end{table}
\section{Evaluation}

\subsection{Experiment Setting}

\subsubsection{Environment} We perform all experiments on a single Linux server equipped with 20-core 2.4 GHz Intel Xeon CPUs and 1T RAM. We used a single process to simulate a single party. Additionally, we apply tc tool\footnote{\url{https://man7.org/linux/man-pages/man8/tc.8.html}} to simulate LAN setting with a bandwidth of 1GBps and sub-millisecond round-trip time (RTT) latency and  WAN setting with 40Mbps bandwidth and 40ms RTT latency, following prior work~\cite{mohassel2018aby3}. Without specifications, the default network is the WAN setting.

\subsubsection{Implementation}
We implement \SecVPre in C++, integrating essential components such as OPRF, OKVS, cuckoo hashing, and simple hashing from~\cite{rindal2022blazing}. In particular, for OPRF and OKVS, we leverage the optimized 3H-GCT algorithm from~\cite{rindal2022blazing}, setting the cluster size to $2^{14}$ and the weight parameter to $3$.
For hashing, we configure the hash table size as $b = 1.27m$ and utilize $h = 3$ hash functions. Additionally, we set the bit-length parameter to $l = 64$ and define the computational ring as $\mathbb{Z}_{2^{64}}$. The protocol operates with a computational security parameter of $\kappa = 128$ and statistical security parameter of $\lambda = 40$. The source code of \SecVPre is publicly available\footnote{\SecVPre is available at: https://github.com/chenshuyuhhh/IDCloak.git.}.

\subsubsection{Baseline}
To the best of our knowledge, \baseline is the only framework with identical functionality to \SecVPre, except that it is limited to the two-party setting. 
\baseline, like our scheme, first computes the intersection IDs, aligns the features, and uses shuffle to remove redundant data. Since \baseline is not open-sourced, we re-implement it\footnote{\baseline is available at: https://github.com/chenshuyuhhh/iPrivJoin.git.}.
Both \baseline and \SecVPre are evaluated under identical experimental settings, except for the number of parties $n=2$ for \baseline.
We do not include mPSI-based dataset join methods as our baselines because they expose intersection IDs in plaintext, whereas \SecVPre keeps them private for stronger privacy protection.

\subsubsection{Datasets}
As is shown in Table~\ref{tab:datasets}, we employ six widely-used real-world datasets, \textit{phishing}, \textit{myocardial infarction complications}, \textit{Appliances energy prediction}, \textit{bank marketing}, \textit{Give me some credit}, \textit{Health Indicators}
all of which are available in the UCI Machine Learning Repository~\cite{asuncion2007uci} or Kaggle~\cite{GiveMeSomeCredit}. These datasets have 10 to 111 features and contain between 1353 and 253680 samples. We partition each dataset vertically, ensuring that all parties hold the same number of samples as in the original dataset, but with uniformly distributed disjoint feature subsets. Specifically, $80\%$ of the samples retain randomly selected intersection IDs, whereas the remaining $20\%$ consist of randomly generated samples.

\begin{table}[htbp]
\belowrulesep=0pt 
\aboverulesep=0pt
\centering
\caption{Datasets used for end-to-end experiments}
\scalebox{0.85}{
\begin{tabular}{l|c|c}
\toprule
Datasets & \#Features ($d$) & Dataset size ($m$) \\
\midrule
 Phishing & 10 & 1353  \\
 Myocardial & 111 & 1700 \\
 Energy & 29 & 19735 \\
 Bank & 17 & 45211 \\
 Credit & 12 & 150000 \\
 Health& 21 & 253680 \\
\bottomrule
\end{tabular}
}
\label{tab:datasets}
\end{table}

\subsection{Evaluation of \SecVPre on Real-world Datasets}

We first conduct end-to-end online efficiency experiments by varying the number of parties ($n$) from $2$ to $6$. 
As is shown in Table~\ref{tab:e2eeva}, we can summarize the following:
\begin{itemize}[leftmargin=1em,itemsep=0pt, topsep=0pt, partopsep=0pt, parsep=0pt]
    \item In the two-party setting, \SecVPre outperforms \baseline by $1.69\times \sim 1.92\times$ and $1.50\times \sim 1.72\times$ in terms of time and communication sizes, respectively. This efficiency gain primarily stems from \SecVPre using lightweight OKVS, whereas \baseline relies on the more communication-expensive OPPRF (Section \ref{sec:design_comm} provides a detailed comparison of the communication complexity).
    \item In multi-party settings involving $3$ to $6$ parties, \SecVPre maintains efficiency comparable to \baseline. This is attributed to the scalability of our multi-party protocol, which preserves performance as the number of parties increases. 
    Notably, \SecVPre demonstrates linear scalability with respect to the number of parties, highlighting its practicality for multi-party settings.
\end{itemize}

\begin{table*}[htbp]
\belowrulesep=0pt 
\aboverulesep=0pt
\centering
\caption{
\SecVPre vs. the two-party framework \baseline\cite{liu2023iprivjoin} on six real-world datasets.
}
\scalebox{0.87}{
\setlength{\tabcolsep}{4.2pt}
\begin{tabular}{c|c|cccccc|cccccc}
\toprule
\multirow{2}{*}{Framework}             & \multirow{2}{*}{$n$} & \multicolumn{6}{c|}{Time (s)}                              & \multicolumn{6}{c}{Communication sizes (MB)}      
\\ \cline{3-14}
                           &                    & Phishing & Myocardial & Energy & Bank  & Credit & Health & Phishing & Myocardial & Energy & Bank   & Credit & Health  \\ \midrule
\baseline                  & 2                  & 3.13    & 8.36        & 22.82  & 31.62 & 70.61  & 192.53 & 2.88    & 13.13       & 42.14  & 62.35  & 146.46 & 401.53  \\ \midrule
\multirow{5}{*}{\SecVPre}      & 2                  & 1.73    & 4.37        & 12.49  & 16.50 & 41.72  & 109.13 & 1.82    & 8.74        & 25.91  & 36.29  & 89.38  & 240.14  \\
                           & 3                  & 2.05    & 5.61        & 17.99  & 21.68 & 60.74  & 162.79 & 3.82    & 19.96       & 59.42  & 83.06  & 205.11 & 560.17  \\
                           & 4                  & 2.37    & 7.22        & 23.63  & 28.96 & 79.93  & 212.31 & 6.09    & 34.45       & 104.18 & 145.91 & 356.62 & 975.59  \\
                           & 5                  & 2.64    & 8.97        & 29.20  & 37.76 & 99.36  & 265.19 & 8.69    & 52.63       & 159.00 & 223.02 & 540.55 & 1502.14 \\
                           & 6                  & 3.02    & 10.85       & 35.49  & 47.59 & 120.60 & 328.47 & 11.41   & 74.42       & 225.85 & 314.31 & 772.74 & 2129.98
                      \\ \bottomrule
\end{tabular}
}
\label{tab:e2eeva}
\end{table*}

\subsection{Evaluation of cmPSI Protocol}

\subsubsection{Evaluation of Optimization for cmPSI Protocol}
We conduct experiments to evaluate the performance of our proposed cmPSI protocol using an optimized communication structure, comparing it against ring-based and star-based communication structures under varying network conditions of 10Mbps, 40Mbps, and 100Mbps bandwidth, with a fixed network latency of 40ms.
As is shown in Figure~\ref{fig:psiopt}, the cmPSI protocol with optimized communication structure outperforms the ring-based communication structure and the star-based communication structure by $1.14\times \sim 1.80\times$ and $1.14 \times \sim 1.41\times$, respectively. 
Furthermore, the performance advantage of the optimized communication structure becomes increasingly significant as the number of parties increases.

\begin{figure}[htbp]
    \centering
    \captionsetup[subfigure]{labelfont={scriptsize}, textfont={scriptsize}}
    \subfloat[\scriptsize 10Mbps bandwidth\label{fig:psi_net10}]{%
        \includegraphics[width=0.15\textwidth]{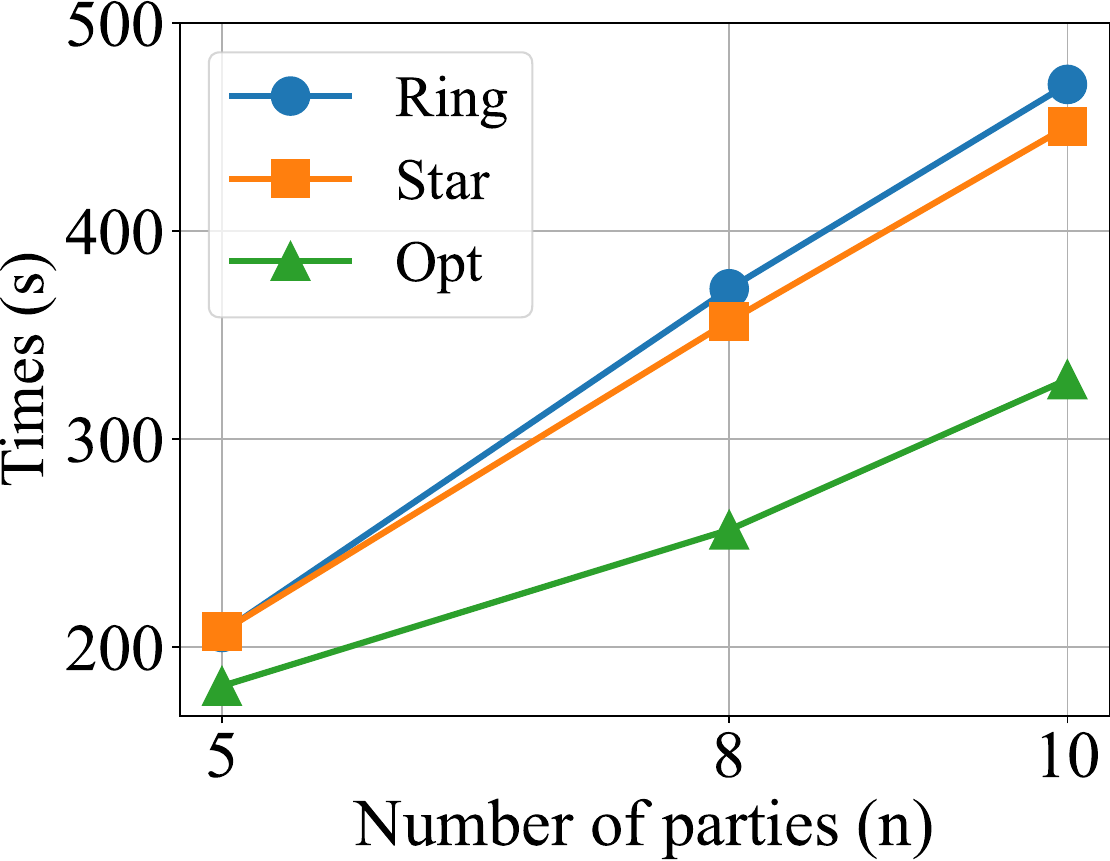}%
    }
    \hspace{0.3em}
    \subfloat[\scriptsize 40Mbps bandwidth\label{fig:psi_net40}]{%
        \includegraphics[width=0.15\textwidth]{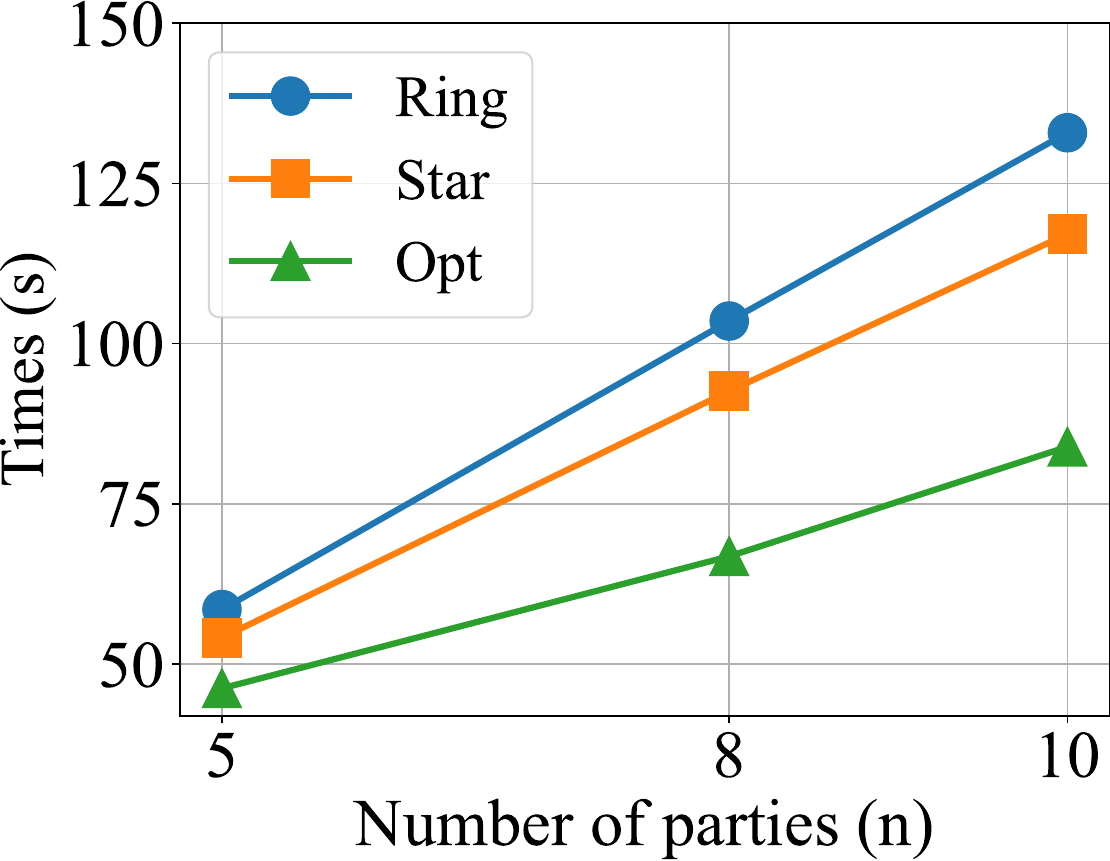}%
    }
    \hspace{0.3em}
    \subfloat[\scriptsize 100Mbps bandwidth\label{fig:psi_net100}]{%
        \includegraphics[width=0.145\textwidth]{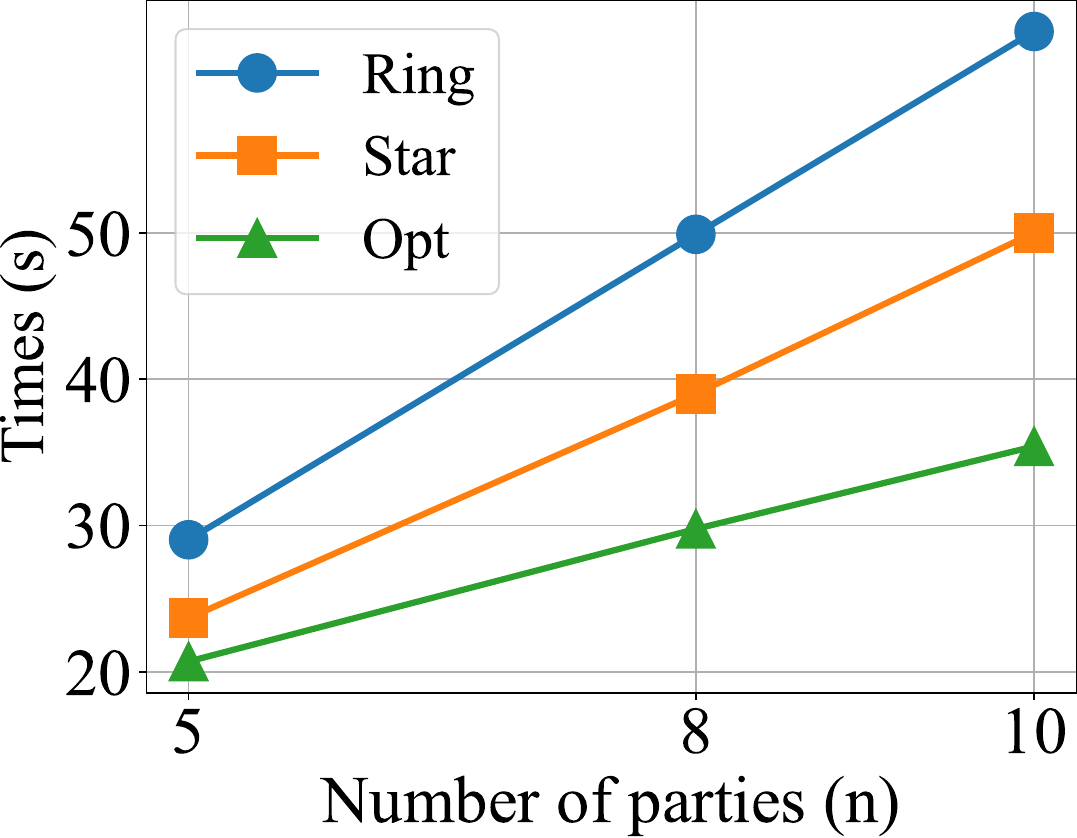}%
    }
    \caption{Comparison of time (in seconds) for different communication structures in cmPSI protocol under different network settings. `Ring', `Star', and `Opt' denote a ring-based communication structure, a star-based communication structure, and optimized communication structure, respectively. }
    \label{fig:psiopt}
\end{figure}

\subsubsection{Evaluation of SOTA Comparison for cmPSI Protocol}
We conduct experiments to compare our proposed cmPSI protocol with the SOTA cmPSI protocol~\cite{chandran2021efficient}, which also keeps the intersection IDs private. Our experiments configure each protocol to resist the maximum number of colluding parties that it can tolerate. Specifically, our proposed cmPSI protocol provides stronger security  (resists against up to \(n-1\) colluding parties), whereas~\cite{chandran2021efficient} only resists against up to \(n/2\) colluding parties.
Moreover, to align its output representation with~\cite{chandran2021efficient}, we add a secure zero-equality operator~\cite{mp-spdz} to the output of our proposed cmPSI protocol.
Specifically, in our protocol, when an ID in \(P_1\)'s hash bin belongs to the intersection, all other parties receive a secret-shared flag of 0; otherwise, they receive a secret-shared flag of random values. In contrast, \cite{chandran2021efficient} assigns a secret-shared flag of 1 indicating intersection IDs.

\begin{table}[htbp]
\centering
\belowrulesep=0pt 
\aboverulesep=0pt
 \caption{
Our cmPSI vs. SOTA cmPSI protocol\cite{chandran2021efficient} on datasets with sizes $m \in \{ 2^{16}, 2^{20} \}$, and involving 3 to 10 parties. `Comm' is short for communication.}
\label{tab:cmpsi-eva}
\scalebox{0.82}{
\setlength{\tabcolsep}{5pt}
\begin{tabular}{c|c|c|cccc}
\toprule
              & $m$                   & Protocol       & 3       & 5       & 8       & 10      \\ \midrule
\multirow{4}{*}{\begin{tabular}{@{}c@{}}LAN \\ Time \\ (s)\end{tabular}}  & \multirow{2}{*}{$2^{16}$} & \cite{chandran2021efficient} & 1.46    & 2.04    & 2.89    & 3.47    \\
                      &                     & Our cmPSI           & 1.05    & 1.43    & 1.99    & 2.34    \\ \cline{2-7}
                      & \multirow{2}{*}{$2^{20}$} & \cite{chandran2021efficient} & 22.40   & 33.62   & 48.12   & 57.60   \\
                      &                     & Our cmPSI           & 16.09   & 19.16   & 26.16   & 28.35   \\ \midrule
\multirow{4}{*}{\begin{tabular}{@{}c@{}}WAN \\ Time \\ (s)\end{tabular}}  & \multirow{2}{*}{$2^{16}$} & \cite{chandran2021efficient} & 16.74   & 30.78   & 49.87   & 63.72   \\
                      &                     & Our cmPSI           & 6.60    & 8.47    & 11.37   & 13.08   \\ \cline{2-7}
                      & \multirow{2}{*}{$2^{20}$} & \cite{chandran2021efficient} & 278.95  & 518.67  & 875.99  & 1130.24 \\
                      &                     & Our cmPSI           & 56.27   & 84.90   & 118.53  & 145.25  \\ \midrule
\multirow{4}{*}{\begin{tabular}{@{}c@{}}Comm \\ Sizes \\ (MB)\end{tabular}} & \multirow{2}{*}{$2^{16}$} & \cite{chandran2021efficient} & 82.09   & 171.49  & 312.14  & 414.87  \\
                      &                     & Our cmPSI           & 14.72   & 29.44   & 51.52   & 66.25   \\ \cline{2-7}
                      & \multirow{2}{*}{$2^{20}$} & \cite{chandran2021efficient} & 1412.08 & 2941.12 & 5339.59 & 8158.22 \\
                      &                     & Our cmPSI           & 207.66  & 415.31  & 726.80  & 934.45 
\\ \bottomrule
\end{tabular}
}
\end{table}

As is shown in Table~\ref{tab:cmpsi-eva}, we can summarize the following:
Our proposed cmPSI protocol, which offers stronger security, outperforms SOTA cmPSI~\cite{chandran2021efficient} by $2.54\times \sim 7.78\times$ and $1.39\times \sim 2.03\times$ in the WAN and LAN settings, respectively, in terms of time. These improvements primarily stem from the reduced communication overhead in our protocol, which achieves a $5.58\times \sim 8.73\times$ reduction compared to cmPSI~\cite{chandran2021efficient}. Specifically,  cmPSI~\cite{chandran2021efficient} relies on PSM primitives, which impose higher communication costs than the OKVS and OPRF primitives used in our protocol. In addition, whereas the SOTA cmPSI protocol employs a star communication structure, our proposed protocol adopts an optimized communication structure that enhances the efficiency without compromising security. Notably, as the number of parties $n$ increases or the dataset size $m$ grows, the advantages of our protocol become even more pronounced.

\subsection{Evaluation of smFA Protocol}

\subsubsection{Evaluation of SOTA Comparison for Secure Shuffle Protocol}
We evaluate the efficiency of our proposed secure multi-party shuffle protocol, a core component of the smFA protocol, against the SOTA secure multi-party shuffle protocol implemented by MP-SPDZ~\cite{mp-spdz}.  
As shown in Table~\ref{tab:shuffle-eva}, our proposed secure shuffle protocol achieves $21.44 \times \sim 41.47\times$ and $61.62\times \sim 138.34\times$ speedup over the SOTA protocol~\cite{mp-spdz} in the LAN and WAN settings, respectively. These improvements primarily stem from significantly lower communication sizes and rounds, precisely, a $105.69\times \sim 132.13\times$ reduction in communication sizes compared to ~\cite{mp-spdz}. 
Specifically, our proposed secure shuffle protocol incurs the communication sizes of \( \mathcal{O}(ndlm) \) per party over \( n \) communication rounds, whereas the SOTA secure shuffle protocol~\cite{mp-spdz} leverages Waksman networks~\cite{waksman1968permutation} for shuffling, resulting in larger communication sizes of \( \mathcal{O}(ndlm\log m) \) and a requirement of \( \mathcal{O}(\log m) \) rounds per party. Consequently, our protocol demonstrates a more pronounced efficiency advantage as the number of parties and the dataset size increase. 
For instance, in the WAN setting with \( n=5 \), increasing \( m \) from \( 2^{16} \) to \( 2^{20} \) results in an efficiency improvement over \cite{mp-spdz} from $104.50\times$ to $138.34\times$. Similarly, when \( m=2^{20} \), increasing \( n \) from 3 to 5  boosts the efficiency improvement from $81.89\times$ to $138.34\times$.

\begin{table}[htbp]
\belowrulesep=0pt 
\aboverulesep=0pt
\caption{The online Time (in seconds) and communication size (in MBs) of our proposed secure multi-party shuffle protocol vs. SOTA secure multi-party shuffle protocol\cite{mp-spdz} on datasets with sizes $m \in \{ 2^{16}, 2^{20} \}$, and involving 3 and 5 parties. `Pro' and `Comm' are short for protocol and communication. respectively.}
\label{tab:shuffle-eva}
\setlength{\tabcolsep}{1.5pt}
\scalebox{0.83}{
\begin{tabular}{c|c|cc|cc|cc}
\toprule
\multirow{2}{*}{$m$}  & \multirow{2}{*}{Pro} & \multicolumn{2}{c|}{LAN Time (s)} & \multicolumn{2}{c|}{WAN Time (s)} & \multicolumn{2}{c}{Comm sizes (MB)} \\ \cline{3-8}
                    &                            & 3               & 5              & 3              & 5               & 3                    & 5                     \\ \midrule
\multirow{2}{*}{$2^{16}$} & \cite{mp-spdz}                     & 38.79           & 144.37         & 2154.36        & 11204.60        & 12484.92             & 68456.40              \\
                    & Ours                        & 1.81            & 6.49           & 34.96          & 107.22          & 118.11               & 647.70                \\ \midrule
\multirow{2}{*}{$2^{20}$} & \cite{mp-spdz}                     & 1330.87         & 3316.92        & 46267.61       & 239100.02       & 249687.80            & 1369105.00            \\
                    & Ours                        & 32.09           & 113.24         & 564.98         & 1728.33         & 1889.76              & 10363.20      
\\ \bottomrule
\end{tabular}
}
\end{table}

\subsubsection{Evaluation of Efficiency Improvement for smFA using our secure shuffle}
We evaluate the efficiency improvements of the smFA protocol from using the SOTA secure multi-party shuffle provided by MP-SPDZ~\cite{mp-spdz} to use our proposed secure shuffle protocol.
As is shown in Figure~\ref{fig:fa_shu}, our proposed secure shuffle protocol significantly reduces the time and communication sizes of the smFA protocol by $40.54\times \sim 53.58\times$ and $65.63\times \sim 81.29\times$ respectively, compared with integrating the SOTA secure shuffle in the smFA protocol. 
The improvements occur because the shuffle operation alone contributes to over \(99\%\) of the total time and communication sizes when using the SOTA secure multi-party shuffle protocol. 

\begin{figure}[htbp]
    \centering
    \captionsetup[subfigure]{labelfont={scriptsize}, textfont={scriptsize}}
    \subfloat[{\scriptsize Time ($m=2^{16}$)}]{\includegraphics[width=0.22\textwidth]{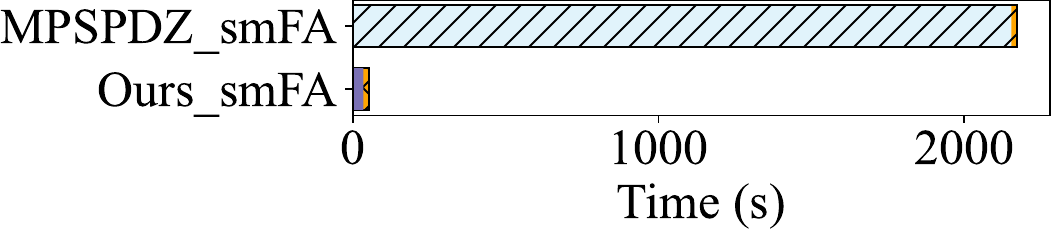}}
    \hspace{0.15cm}
    \subfloat[\scriptsize Communication ($m=2^{16}$)]{\includegraphics[width=0.22\textwidth]{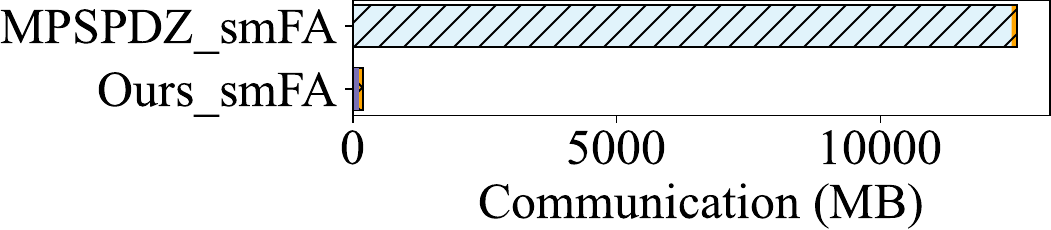}}

 \subfloat[\scriptsize Time ($m=2^{20}$)]{\includegraphics[width=0.22\textwidth]{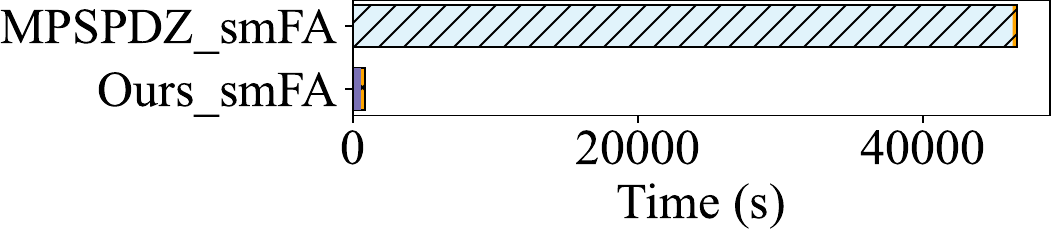}}
    \hspace{0.15cm}
    \subfloat[\scriptsize Communication ($m=2^{20}$)]{\includegraphics[width=0.22\textwidth]{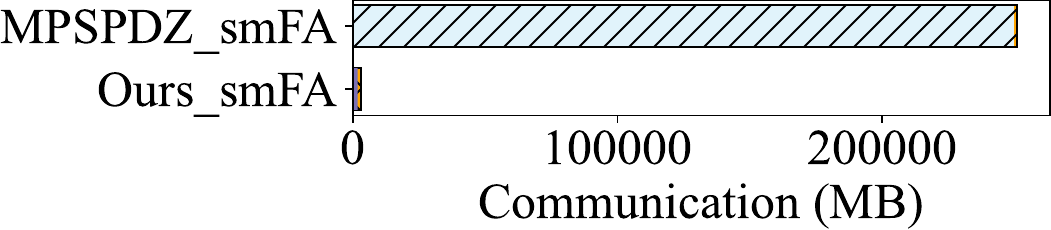}}

\subfloat{
\includegraphics[width=0.49\textwidth]{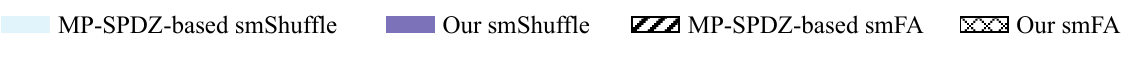}}
\caption{Online time (in seconds) and communication sizes (in MBs) of the smFA protocol using our proposed secure multi-party shuffle vs. the smFA protocol using the SOTA secure multi-party shuffle~\cite{mp-spdz}. Evaluations are conducted on datasets with varying dataset sizes \( m \in \{2^{16}, 2^{20}\} \), fixed feature dimension \( d_i = 10 \)(\( i \in [n] \)), and a fixed number of parties \( n = 3 \).}
\label{fig:fa_shu}
\end{figure}

\subsubsection{Evaluation of Efficiency for smFA Protocol}
We evaluate the efficiency of our proposed smFA protocol under varying numbers of parties ($n$), dataset sizes ($m$), and feature dimensions ($d$). 
As is shown in Figure~\ref{fig:fa}, the time and communication sizes of the smFA protocol exhibit linear growth with respect to the number of parties, dataset size, and feature dimensions. These results highlight the practicality of the smFA protocol in multi-party settings.

\begin{figure}[htbp]
    \centering
    \captionsetup[subfigure]{labelfont={scriptsize}, textfont={scriptsize}}
    \subfloat[{\scriptsize Parties numbers}\label{fig:2}]{\includegraphics[width=0.155\textwidth]{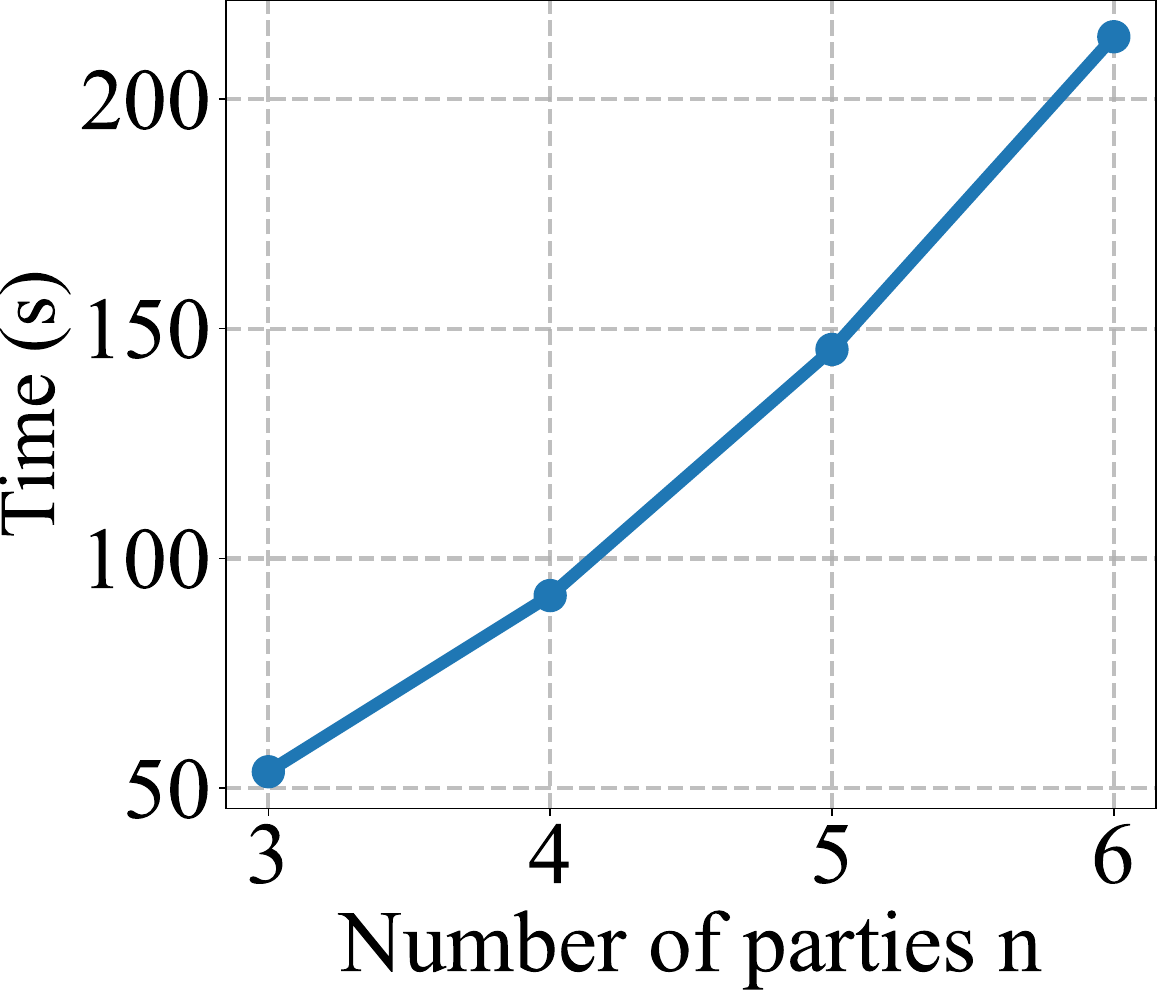}}
    \hspace{0.1cm}
    \subfloat[\scriptsize Dataset sizes \label{fig:1}]{\includegraphics[width=0.152\textwidth]{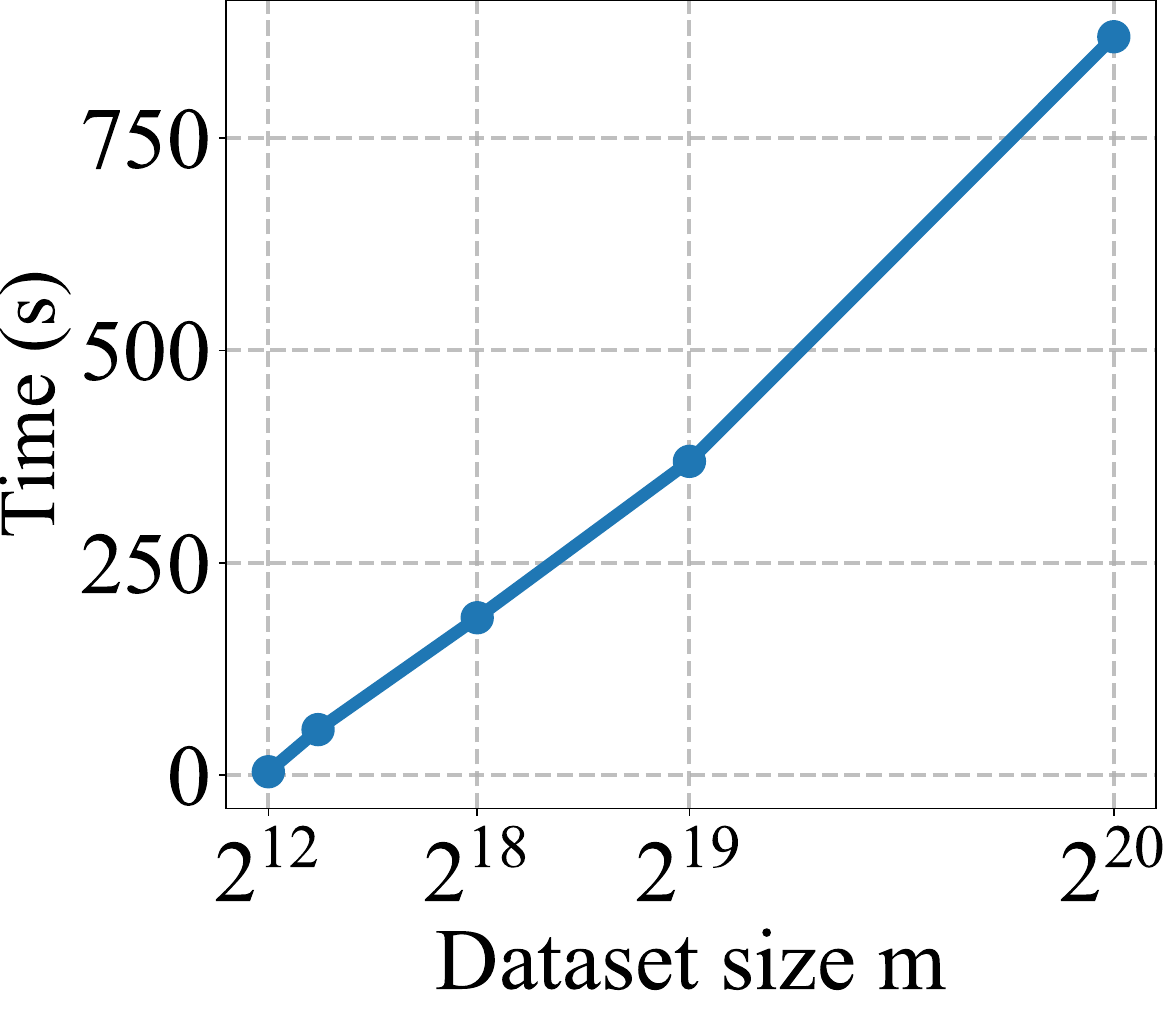}}
    \hspace{0.1cm}
    \subfloat[\scriptsize Feature Dimensions\label{fig:3}]{\includegraphics[width=0.15\textwidth]{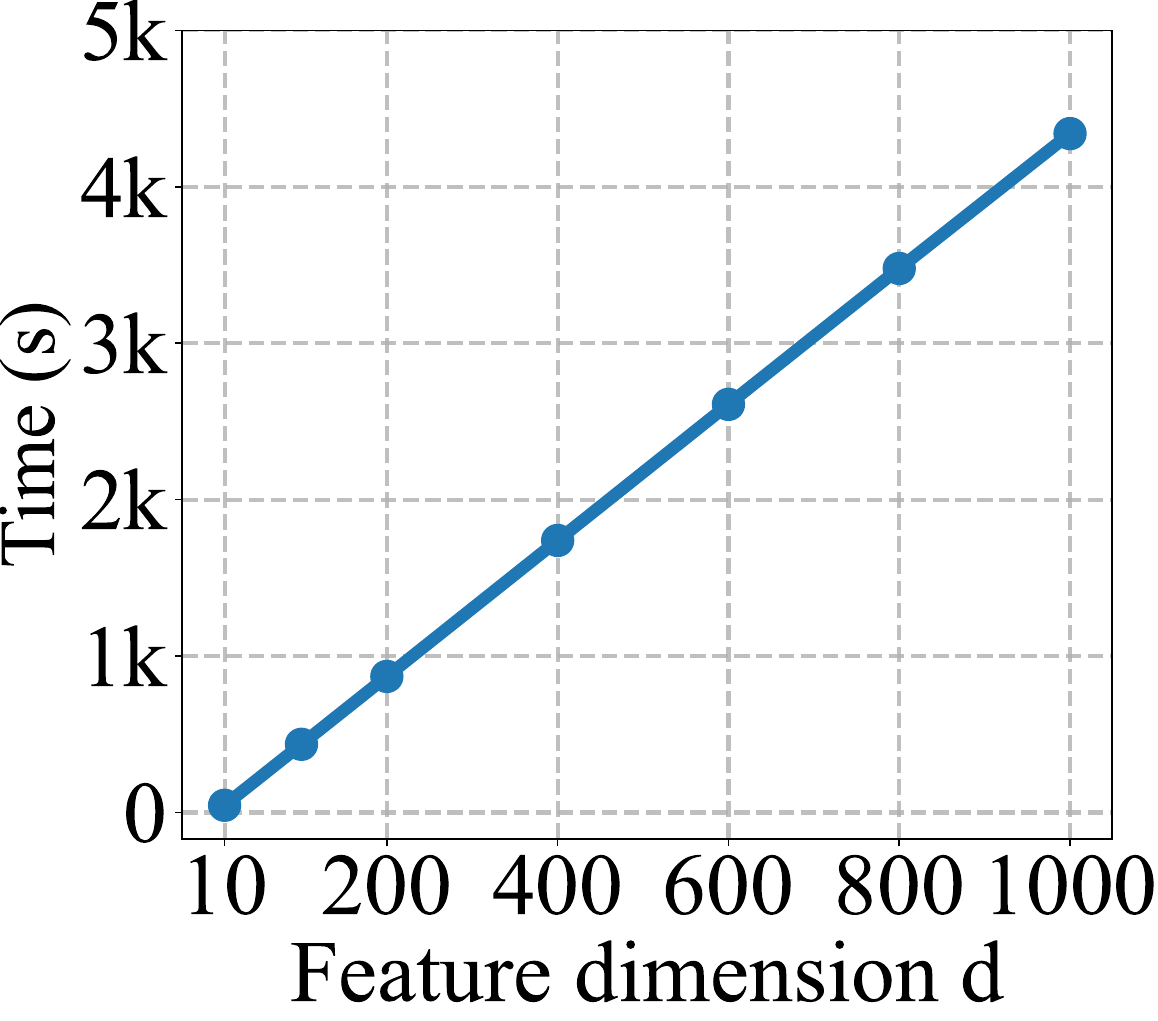}}

    \subfloat[\scriptsize Parties numbers\label{fig:5}]{\includegraphics[width=0.155\textwidth]{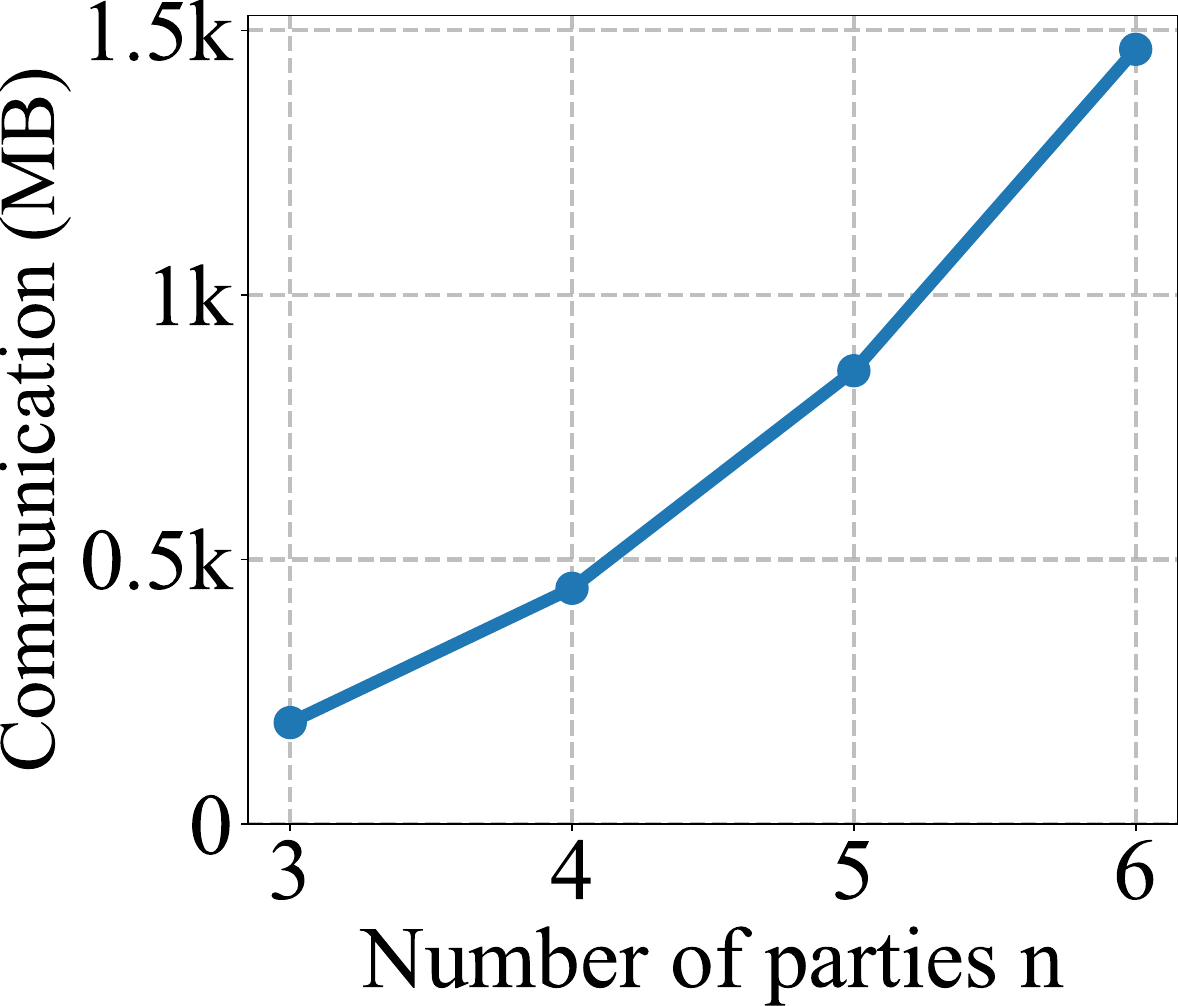}}
    \hspace{0.1cm}
    \subfloat[\scriptsize Dataset sizes\label{fig:4}]{\includegraphics[width=0.145\textwidth]{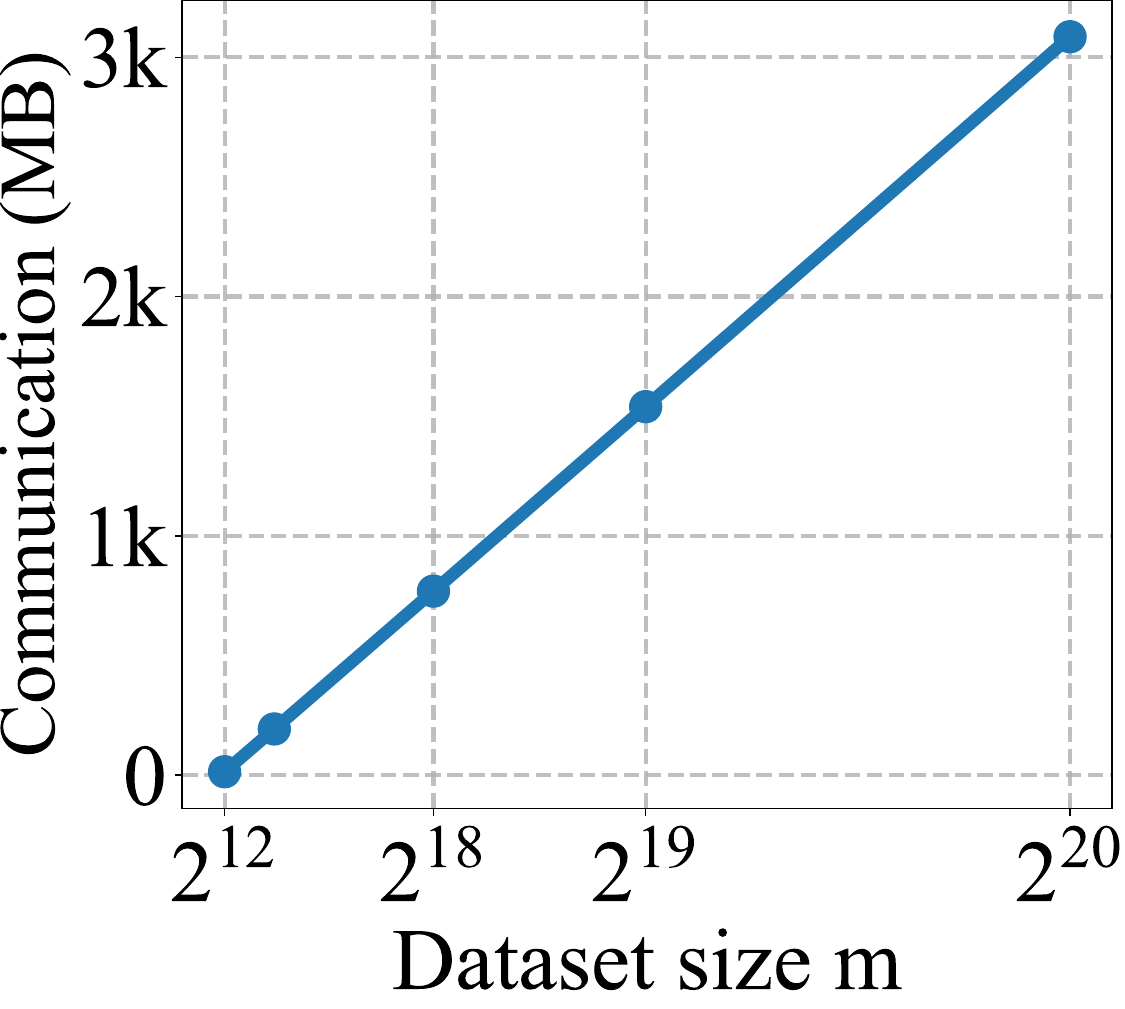}}
    \hspace{0.1cm}
    \subfloat[\scriptsize Feature Dimensions\label{fig:6}]{\includegraphics[width=0.155\textwidth]{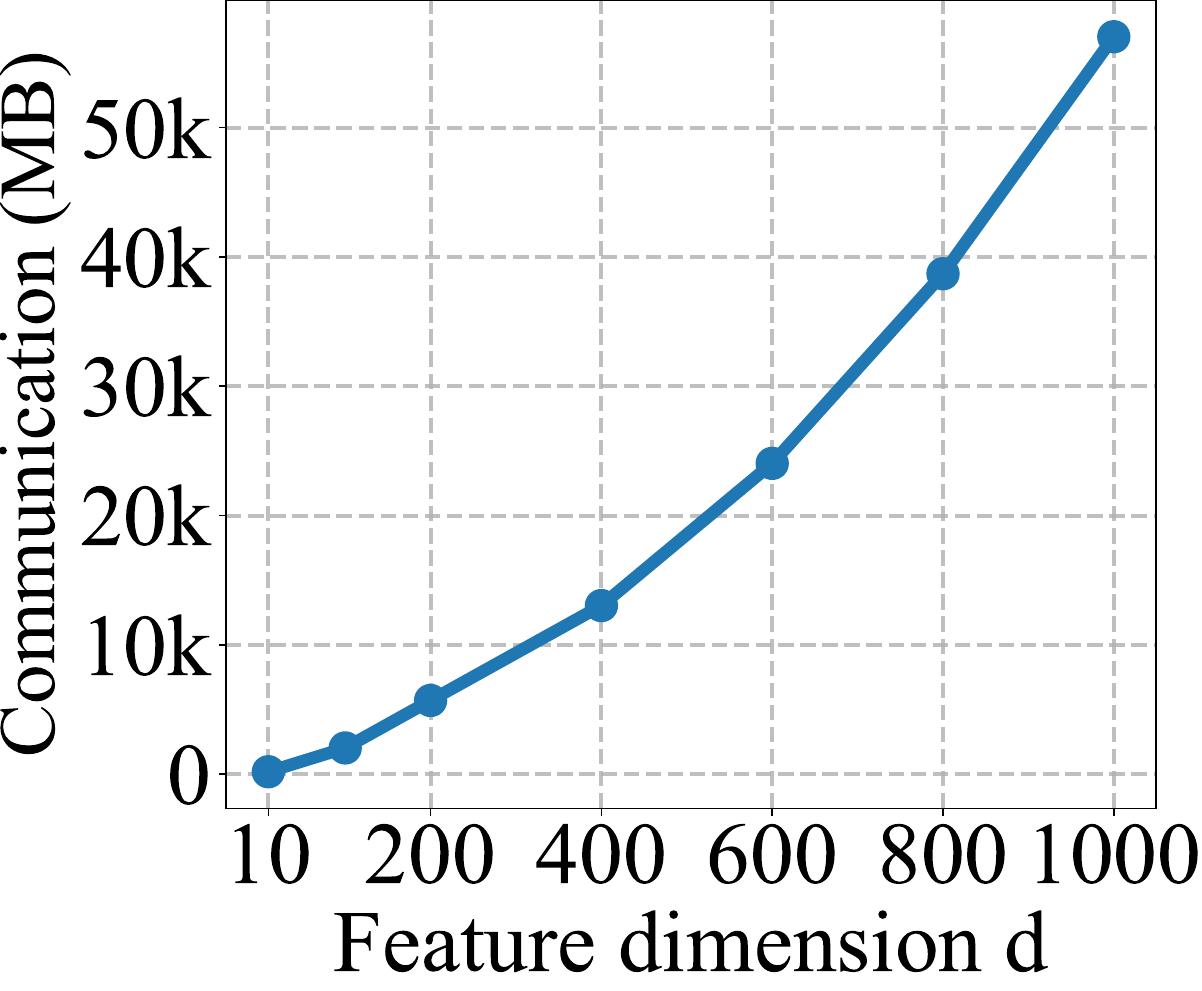}}
    \centering
    \caption{Online Time (in seconds) and communication sizes (in MB) of our smFA protocol under varying parameters. When varying \( n \), parameters \( d_i = 10 \) and \( m = 2^{16} \) are fixed; when varying \( m \), parameters \( n = 3 \) and \( d_i = 10 \) are fixed; when varying \( d_i \), parameters \( n = 3 \) and \( m = 2^{16} \) are fixed. Note: `k' denotes 1000.}
    \label{fig:fa}
\end{figure}
\section{Discussion}

\noindent \textbf{Revealing the Size of the Intersection Set.} 
Although \SecVPre reveals the intersection size $c$, this information is essential for practical applications \cite{ ghosh2019communication,badrinarayanan2021multi}. In \VPPML, parties often rely on \(c\) to decide whether to continue training; if \(c\) is below a certain threshold (e.g., \(c < 50\% \cdot n\)), they can terminate early to avoid wasting effort. Thus, revealing \(c\) in \SecVPre is reasonable. Moreover, \SecVPre provides \(c\) in the cmPSI protocol, allowing parties to end the process promptly if the intersection size is too small.

\noindent \textbf{Supporting \VPPML Frameworks with Different Schemes.} 
Our \SecVPre generates training dataset in ASS, a commonly used scheme in \VPPML frameworks \cite{mohassel2017secureml,mp-spdz}. However, practical applications may require other schemes, such as boolean secret sharing or homomorphic encryption, depending on the training setup. 
To address this, \SecVPre can use existing conversion protocols \cite{demmler2015aby,mohassel2018aby3} to transform ASS dataset into the required scheme.
\section{Conclusion}

In this paper, we investigate that existing frameworks for secure dataset join in \VPPML could be impractical because they are insecure when they expose the intersection IDs; or they rely on a strong trust assumption requiring a non-colluding auxiliary server; or they are limited to the two-party setting. 
To resolve the problem, we propose \SecVPre, the first practical secure multi-party dataset join framework for \VPPML without a non-colluding auxiliary server. \SecVPre consists of two efficient protocols.
First, our proposed circuit-based multi-party private set intersection (cmPSI) protocol securely computes secret-shared flags indicating intersection IDs. Second, our proposed secure multi-party feature alignment protocol leverages our proposed secure multi-party shuffle protocol to construct the secret-shared and joint dataset based on these secret-shared flags. 
Our experiments show that in the two-party setting, \SecVPre outperforms \baseline, a state-of-the-art secure two-party dataset join framework, and maintains comparable efficiency even as the number of parties increases.
Furthermore, compared to the SOTA protocol, our proposed cmPSI protocol offers a stronger security guarantee while improving efficiency by up to $7.78\times$ in time and $8.73\times$ in communication sizes.
Additionally, our secure multi-party shuffle protocol outperforms the SOTA protocol by up to $138.34\times$ in time and $132.13\times$ in communication sizes. 

\bibliographystyle{IEEEtran}
\bibliography{citations}

\section*{Appendix}

\renewcommand{\thesubsection}{\Alph{subsection}}
\renewcommand{\thesubsubsection}{\Alph{subsection}.\arabic{subsubsection}}

\subsection{Algorithm for optimized communication structure of cmPSI}
As is shown in algorithm~\ref{alg:optstruct}, it aims to construct an efficient communication structure for our cmPSI protocol, based on the parameters, i.e. the number of parties $n$, the number of designated root party $leader\_num$, the time for sending one OKVS table $t_s$, and the network delay ($t_l$).
The core idea is to maximize node utilization while minimizing communication delay.
Node is a data structure that has the num, children, and parent, where the `num' indicates the identifier of a party, the `parent' indicates the parties to whom this Node will send data, the `child' indicates the parties whose data this Node will receive.
The algorithm begins by computing $k$, which determines the number of children each node should have in the structure. Nodes are then grouped into sets of size $k$, with the last node in each group acting as the parent for the others. This grouping process is iterated until a single root node remains, which represents the leader. If the resulting root does not match the designated leader, a breadth-first search is used to locate and swap in the correct leader node.

\begin{algorithm}[H]
\caption{Optimized communication structure generation}
\setstretch{0.8}
\label{alg:optstruct}
\begin{algorithmic}[1]
\REQUIRE The number of parties $n$, the number of root party $leader\_num$, the time for sending one OKVS table $t_s$, the time of network delay $t_l$.
\ENSURE Root node $node$ of the communication structure.

\STATE $k = \lceil \frac{t_l}{t_s} \rceil + 1$, $c = n \bmod k$, $nodes,nodes2 = \emptyset$
  \FOR{$i = 0$ to $n - 1$}
    \STATE $nodes.append(Node(i))$
  \ENDFOR
  \WHILE{$n > 1$}
    \STATE $c = n \bmod k$ 
    \FOR{$i = 0$ to $n-k-c$ step $k$} 
        \FOR{$j = 0$ to $k - 1$} 
          \STATE $nodes[i + k - 1].add\_child(nodes[i + j])$
          \STATE $nodes[i + j].parent = nodes[i + k - 1]$
        \ENDFOR
        \STATE $nodes2.append(nodes[i+k-1])$
    \ENDFOR
    \FOR{$i = n-k-c$ to $n-1$} 
        \STATE $nodes[n-1].add\_child(nodes[i])$
        \STATE $nodes[i].parent = nodes[n-1]$
    \ENDFOR
    \STATE $nodes2.append(nodes[n-1])$
    \STATE $nodes = nodes2$, $nodes2 = \emptyset$
    \STATE $n=len(nodes)$
    \ENDWHILE
\STATE $node = nodes[0]$
\IF{$node.num == leader\_num$}
     \STATE Return
\ENDIF
\STATE   $queue = \text{deque}([node])$, $t\_node = \text{None}$
\WHILE{$queue \neq \emptyset$}
\STATE     $current = queue.popleft()$
    \IF{$current.num == leader\_num$}
       \STATE $t\_node = current$; Break
    \ENDIF
    \FOR{$child \in current.child$} 
       \STATE $queue.append(child)$
    \ENDFOR
\ENDWHILE
\IF{$t\_node \neq \text{None}$}
    \STATE $node.num, t\_node.num = t\_node.num, node.num$
\ENDIF
\STATE Return $node$
\end{algorithmic}
\label{alg:opt}
\end{algorithm}

\vfill

\end{document}